\documentclass[11pt]{article}
\usepackage[abbrvbib,hyperref]{arxiv-template}

\usepackage[T1]{fontenc}
\usepackage[utf8]{inputenc}
\usepackage[english]{babel}

\usepackage{lineno}

\usepackage{subcaption}

\newdimen\figwidth

\usepackage[thmmarks]{ntheorem}

\theoremheaderfont{\normalfont\bfseries}
\theorembodyfont{\normalfont}
\theoremseparator{.}
\theoremstyle{plain}

\newtheorem{theorem}{Theorem}[section]
\newtheorem{definition}{Definition}

\newtheorem{assumption}{Assumption}

\newtheorem{example}{Example}

\makeatletter

\newcommand{\sharetheoremcounter}[1]{%
  \expandafter\let\csname c@#1\endcsname\c@theorem
  \expandafter\let\csname the#1\endcsname\thetheorem
  \expandafter\let\csname p@#1\endcsname\p@theorem
}

\sharetheoremcounter{definition}
\sharetheoremcounter{lemma}
\sharetheoremcounter{setting}
\sharetheoremcounter{remark}
\sharetheoremcounter{assumption}
\sharetheoremcounter{proposition}
\sharetheoremcounter{corollary}
\sharetheoremcounter{example}

\makeatother

\theoremstyle{nonumberplain}
\theoremheaderfont{\normalfont\bfseries}
\theorembodyfont{\normalfont}
\theoremsymbol{\BlackBox}
\theoremseparator{.}

\makeatletter
\renewtheoremstyle{nonumberplain}%
  {%
    \item[
      \theorem@headerfont
      \hskip\labelsep
      ##1%
      \theorem@separator
    ]%
  }%
  {%
    \item[
      \theorem@headerfont
      \hskip\labelsep
      ##1\ ##3%
      \theorem@separator
    ]%
  }
\makeatother

\newtheorem{proof}{Proof}

\usepackage{thmtools}
\usepackage{thm-restate}

\newif\ifarxiv
\arxivtrue

\newif\ifsupplementary
\supplementaryfalse

\usepackage{tocloft}

\newcommand{\appendixtocname}{%
  Contents of the Appendix%
}

\newlistof{appendices}{app}{\appendixtocname}

\newcommand{\appsection}[1]{%
  \section{#1}%
  \addcontentsline{app}{section}{%
    \protect\numberline{\thesection}#1%
  }%
}

\usepackage[hyperpageref]{backref}

\renewcommand*{\backref}[1]{}

\renewcommand*{\backrefalt}[4]{%
  {\small
    \ifcase #1
      (Not cited.)%
    \or
      (Cited on page~#2.)%
    \else
      (Cited on pages~#2.)%
    \fi
  }%
}
\usepackage{amsmath}
\usepackage{mathtools}
\mathtoolsset{showonlyrefs}

\usepackage{newtxtext}
\usepackage[subscriptcorrection]{newtxmath}

\usepackage{booktabs}
\usepackage{array}
\usepackage{tabularx}
\usepackage{enumitem}

\usepackage{tikz}
\usetikzlibrary{arrows.meta}

\tikzset{
  every node/.style={
    circle,
    draw,
    minimum size=9mm,
    inner sep=0pt
  },
  >={To[length=4pt,width=6pt]}
}

\tikzset{
  highlight/.style={
    preaction={
      draw,
      gray!30,
      line width=12pt,
      line cap=round,
      line join=round,
      shorten <=0.5\pgflinewidth,
      shorten >=0.5\pgflinewidth,
      -
    }
  }
}

\usepackage{zref}
\usepackage{zref-clever}

\zcsetup{
  rangesep={--}
}

\zcRefTypeSetup{theorem}{
  Name-sg = Theorem,
  name-sg = theorem,
  Name-pl = Theorems,
  name-pl = theorems,
}

\zcRefTypeSetup{definition}{
  Name-sg = Definition,
  name-sg = definition,
  Name-pl = Definitions,
  name-pl = definitions,
}

\zcRefTypeSetup{lemma}{
  Name-sg = Lemma,
  name-sg = lemma,
  Name-pl = Lemmas,
  name-pl = lemmas,
}

\zcRefTypeSetup{setting}{
  Name-sg = Setting,
  name-sg = setting,
  Name-pl = Settings,
  name-pl = settings,
}

\zcRefTypeSetup{remark}{
  Name-sg = Remark,
  name-sg = remark,
  Name-pl = Remarks,
  name-pl = remarks,
}

\zcRefTypeSetup{assumption}{
  Name-sg = Assumption,
  name-sg = assumption,
  Name-pl = Assumptions,
  name-pl = assumptions,
}

\zcRefTypeSetup{proposition}{
  Name-sg = Proposition,
  name-sg = proposition,
  Name-pl = Propositions,
  name-pl = propositions,
}

\zcRefTypeSetup{corollary}{
  Name-sg = Corollary,
  name-sg = corollary,
  Name-pl = Corollaries,
  name-pl = corollaries,
}

\zcRefTypeSetup{example}{
  Name-sg = Example,
  name-sg = example,
  Name-pl = Examples,
  name-pl = examples,
}

\usepackage[plain,noend]{algorithm2e}

\makeatletter
\renewcommand{\algocf@captiontext}[2]{%
  #1\algocf@typo. \AlCapFnt{}#2%
}

\def\@algocf@capt@plain{top}

\renewcommand{\algocf@makecaption}[2]{%
  \addtolength{\hsize}{\algomargin}%
  \sbox\@tempboxa{\algocf@captiontext{#1}{#2}}%
  \ifdim\wd\@tempboxa>\hsize
    \hskip .5\algomargin%
    \parbox[t]{\hsize}{%
      \algocf@captiontext{#1}{#2}%
    }%
  \else
    \global\@minipagefalse%
    \hbox to\hsize{\box\@tempboxa}%
  \fi%
  \addtolength{\hsize}{-\algomargin}%
}
\makeatother

\newcounter{subfig}[figure]

\newcommand{\figsubref}[2]{%
  \hyperref[#2]{%
    Figure~\ref*{#1}\ref*{#2}%
  }%
}

\newcommand{\BlackBox}{%
  \begingroup
  \setlength{\fboxsep}{0pt}%
  \fbox{%
    \rule{0pt}{1.5ex}%
    \rule{1.5ex}{0pt}%
  }%
  \endgroup
}

\newcommand{\doop}{%
  \ensuremath{\operatorname{do}}%
}

\usepackage{silence}
\usepackage{microtype}

\newcommand{\stepstrut}{%
  \vphantom{\displaystyle\sum_{t'}}%
}

\title{%
Verifying formulas for interventional distributions
}

\author{
{\name Francesco Freni\textsuperscript{1}}
\and
{\name Leonard Henckel\textsuperscript{2,*}}
\and
{\name Sebastian Weichwald\textsuperscript{1,*}}
\\[0.5em]
\addr \textsuperscript{1}Department of Mathematical Sciences, University of Copenhagen, Denmark
\\
\addr \textsuperscript{2}School of Mathematics and Statistics, University College Dublin, Ireland
\\
\addr
\textsuperscript{*}Equal contribution.
}

\begin{document}
\maketitle

\begin{abstract}
We formalize \emph{verification}
in causal graphical models:
deciding whether a given observational formula identifies a target interventional distribution.
This opens a problem complementary to identification,
asking not whether any identifying formula exists,
but whether the given formula is identifying.
We show that even sound and complete solutions to identification do not solve verification.
We propose a falsifier
as a first practical route forward,
prove that it induces an almost-surely correct verifier
for regular exponential-family
models,
and use the resulting verifier
to develop the gateway test,
which finds
all sets admissible for use in a front-door formula.
\end{abstract}


\section{Introduction}\label{sec:introduction}

We introduce and formalize the problem of \emph{verification} in causal graphical models \citep{Pearl2009}: 
given a graph, treatment and outcome variables, and a candidate observational formula,
decide whether that formula identifies the target interventional distribution.
This opens a problem complementary to identification,
which asks whether the target interventional distribution
is determined by the graph and
observational distribution, 
and, if so, how to express it as an observational formula \citep{Pearl1995}.
This new problem also requires making explicit what is often left
implicit in identification: specifying which observational formulas
are admissible in the first place.
There exists a rich literature on identification,
including graphical criteria 
\citep{Maathuis2015, Perkovic2018}, 
sound and complete algorithms using graphical decompositions and
do-calculus 
\citep{Tian2002, Huang2006, Shpitser2008, Jaber2022, Chen2026},
and extensions to surrogate experiments,
stochastic policies,
and statistical efficiency analysis 
\citep{Bareinboim2012, Correa2020, Witte2020, Henckel2022, Rotnitzky2020}.
Some existing results
can be repurposed to verify formulas in restricted classes,
for example linear instrumental-variable formulas \citep{Henckel2023}
or adjustment formulas \citep{Shpitser2010, Perkovic2018}.
But verification itself
has not been developed as a problem in its own right.

Verification matters for causal graphical modelling.
Conceptually,
proof assistants such as Lean
highlight the value of independently checking
that a proposed mathematical object
has the claimed meaning \citep{deMoura2021};
here, the object is an observational formula
claimed to be identifying for a target interventional distribution.
Practically,
candidate formulas need not be direct outputs of 
a single identification run:
they may be simplified expressions,
outputs of software or human derivations,
formulas transferred from related graphs,
or alternatives expected
to be easier to estimate efficiently \citep{Guo2023}.
This is important to enable
evolvable causal analysis,
where graphs, assumptions, measurements,
and formulas change over time.
The question is then not whether some identifying formula exists,
but whether this particular formula remains correct
for the causal model currently under consideration.
Methodologically,
verification can help
check derivations,
test graphical criteria and conjectures,
expose new graphical implications or identifying formulas,
and support downstream tasks such as comparing graphs
by the identification claims they share \citep{Henckel2024}.

We first show why verification is not a by-product of
the existing identification machinery.
Sound and complete identification algorithms such as ID \citep{Shpitser2006} 
return an identifying formula
when one exists,
but do not decide 
whether a given alternative formula is also identifying.
Direct proof search in the do-calculus proof system does not solve the verification problem either:
fair enumeration of do-calculus and probability-algebra derivations
only semi-decides derivability of
a candidate formula,
terminating with a certificate when a derivation exists
but potentially running forever otherwise (\zcref[S]{thm:semi-decidability}).
We then propose a falsifier
as a first practical route forward.
Rather than searching for a derivation of the candidate formula,
the falsifier searches for disagreement
between the candidate formula and
the target interventional distribution
in sampled graph-compatible models.
For regular conditional exponential-family 
models,
we prove that this induces an almost-surely correct verifier
relative to the chosen parametric family (\zcref[S]{thm:as-correct-verifier}).
Finally, we illustrate what verification enables
by developing the gateway test,
a sound and exhaustively complete procedure for finding all sets 
whose front-door formula identifies the target interventional distribution.
This shows that verification can also characterize identifying strategies. 
\ifarxiv
    Our code is available at 
    \href{https://github.com/francescofreni/hiprof}{github.com/francescofreni/hiprof}.
\fi
We close by outlining open directions
toward a broader study of verification,
including
strengthening falsification from parametric toward non-parametric 
guarantees,
clarifying the limits of do-calculus proof search,
and extending verification beyond equality of interventional formulas.

\section{Causal graphical models and identification}\label{sec:causal-model-identification}
We fix working notation and definitions here,
collecting graphical and causal background in \zcref[S]{app:preliminaries}.
Throughout, let $\mathcal{G}$ be a causal directed acyclic graph over $\mathbf{V}=\mathbf{O}\sqcup\mathbf{L}$,
with observed variables $\mathbf{O}$,
latent variables $\mathbf{L}$,
and disjoint
outcome and intervention node sets $\mathbf{Y},\mathbf{T}\subseteq{\mathbf{O}}$.
Let $\mathcal{G}^p$ be the latent projection of $\mathcal{G}$ onto $\mathbf{O}$ \citep{Verma1990,Richardson2003},
and let $\mathcal{A}(\mathbf{O})$ denote
the class of latent projections (or acyclic directed mixed graphs) over $\mathbf{O}$.

For $\mathbf{A}\subseteq\mathbf{V}$,
set the product sample space $\mathcal{X}_{\mathbf{A}} = \prod_{V\in\mathbf{A}}\mathcal{X}_V$
and let $\mu_{\mathbf{A}}$ be the corresponding product measure.
All distributions considered admit densities
with respect to the relevant $\mu_{\mathbf{A}}$,
and $\mathcal{P}(\mathcal{X}_{\mathbf{A}})$
denotes the class of such densities.
A density $q\in\mathcal{P}(\mathcal{X}_{\mathbf{V}})$
factorizes according to $\mathcal{G}$ if there exist conditional densities
$\{q_{V\mid\operatorname{pa}(V)}:V\in\mathbf{V}\}$,
such that
$q(\mathbf{v}) = \prod_{V\in\mathbf{V}} q_{V\mid\operatorname{pa}(V)}(\mathbf{v}_V\mid \mathbf{v}_{\operatorname{pa}(V)})$
for $\mu_{\mathbf{V}}$-almost every $\mathbf{v}$.
Let $\mathcal{P}_{\mathcal{G}}(\mathcal{X}_{\mathbf{V}})$
denote the class of such full densities.

The causal directed acyclic graph $\mathcal{G}$,
together with this factorization 
and truncated factorizations for interventions,
specifies a non-parametric causal graphical model:
each $q\in\mathcal{P}_{\mathcal{G}}(\mathcal{X}_{\mathbf{V}})$
induces a marginal observational density $q_{\mathbf{O}}$
and a family of interventional densities;
for $\mathbf{t}\in\mathcal{X}_{\mathbf{T}}$,
we write
$q_{\mathbf{Y}\mid\doop(\mathbf{T}=\mathbf{t})}$
for the induced density of
$\mathbf{Y}$ under the hard intervention $\doop(\mathbf{T}=\mathbf{t})$
(see \zcref[S]{app:causal-preliminaries}). 

\begin{definition}[Observational formula]
\label{def:obsformula}
An \emph{observational formula $\phi$ for $\mathbf{Y}$ under intervention on $\mathbf{T}$} is a well-typed, kernel-preserving symbolic expression
with no free variables other than $\mathbf{y}$ and $\mathbf{t}$, generated from
observational marginals and conditionals
by the grammar in 
\zcref[S]{app:grammar}.
The grammar restricts products and quotients
so that formulas have probabilistic
rather than merely algebraic
semantics.
We write
\[
\llbracket \phi \rrbracket:
\mathcal{P}(\mathcal X_{\mathbf O})\times\mathcal{X}_{\mathbf{T}}
\rightharpoonup
\mathcal P(\mathcal X_{\mathbf{Y}})
\]
for the partial functional 
induced by
an admissible formula $\phi$,
and
$\Phi_{\mathbf{Y},\mathbf{T}}$
for the class of such formulas.
\end{definition}

The grammar excludes arbitrary algebraic expressions that need 
not define densities (see \zcref[S]{app:grammar}), and covers
standard identifying formulas, including those derived via
do-calculus and those written using fixing notation
via
their underlying kernel-preserving operations \citep{Richardson2023}.

When unambiguous, we use the standard shorthand that the arguments of
a density symbol determine the corresponding marginal, conditional, or 
interventional density.  
Densities and conditional-density terms are understood up to the usual
almost-everywhere equivalence; pointwise evaluations are taken only where 
the chosen representatives are defined.
In observational formulas,
marginalization sums are shorthand 
for integration with respect to the relevant measures;
primed or subscripted bound variables
are dummy copies
of the corresponding base variable;
for example,
$t'$ and $t_1$ range over $\mathcal{X}_T$
and are integrated with respect to $\mu_T$.

\begin{definition}[Identifying formula]
\label{def:identifying-formula}
An observational formula $\phi$ is 
\emph{identifying for $\mathbf{Y}\mid \doop(\mathbf{T})$ in $\mathcal{G}$}
if,
for every $q\in\mathcal{P}_{\mathcal{G}}(\mathcal{X}_{\mathbf{V}})$
and every $\mathbf{t}\in\mathcal{X}_{\mathbf{T}}$
for which $\llbracket\phi\rrbracket(q_{\mathbf{O}},\mathbf{t})$ is defined,
\begin{equation}\label{eq:identifying-formula}
    \llbracket \phi \rrbracket (q_{\mathbf{O}},\mathbf{t})
    =
    q_{\mathbf{Y}\mid\doop(\mathbf{T}=\mathbf{t})}
    \quad \mu_{\mathbf{Y}}\text{-almost everywhere.}
\end{equation}
We say that $\phi$ is \emph{identifying for $\mathbf{Y}\mid \doop(\mathbf{T})$ 
in $\mathcal{G}$ relative to} $\widetilde{\mathcal{P}}_{\mathcal{G}}(\mathcal{X}_{\mathbf{V}})\subseteq\mathcal{P}_{\mathcal{G}}(\mathcal{X}_{\mathbf{V}})$ 
if the above equality holds for every $q\in\widetilde{\mathcal{P}}_{\mathcal{G}}(\mathcal{X}_{\mathbf{V}})$
and every $\mathbf{t}\in\mathcal{X}_{\mathbf{T}}$.
\end{definition}

\begin{example}[Interpreting an observational formula]
    Consider the graph $X \to T \to Y$ and $X\to Y$.
    The observational formula
    $\phi_{\mathrm{adj}} = \sum_x p(y \mid t, x) p(x)$
    with free variables $y$ and $t$ denotes a functional $\llbracket\phi_{\mathrm{adj}}\rrbracket$ that maps $(q_{\mathbf{O}},t)$ to the density
    $y\mapsto\int q_{Y\mid T,X}(y\mid t,x)q_{X}(x)\,
    d\mu_{X}(x)$
    on $\mathcal{X}_Y$.
    Since $\{X\}$ is a valid adjustment set in this graph, $\phi_{\mathrm{adj}}$ is identifying for $Y\mid \doop(T)$, while the observational formula $p(y\mid t)$ is not.
\end{example}

Graphical identifiability depends only on the latent projection:
$\phi$ is identifying for $\mathbf{Y}\mid\doop(\mathbf{T})$ in $\mathcal{G}^p$
if and only if it is identifying in any,
equivalently every,
full graph $\mathcal{G}$
whose latent projection onto $\mathbf{O}$ is $\mathcal{G}^p$
\citep[Corollary~49]{Richardson2023}.
For $\mathcal{H}\in\{\mathcal{G},\mathcal{G}^p\}$,
we say that $\mathbf{Y}\mid\doop(\mathbf{T})$ is identifiable in $\mathcal{H}$
if and only if there exists $\phi\in\Phi_{\mathbf{Y},\mathbf{T}}$
that is identifying for $\mathbf{Y}\mid\doop(\mathbf{T})$ in $\mathcal{H}$.

Identification asks whether $\mathbf{Y}\mid\doop(\mathbf{T})$ is identifiable and, when it is, provides an 
identifying formula for $\mathbf{Y}\mid\doop(\mathbf{T})$.
We view an identification procedure $\mathcal{I}$
with associated, possibly restricted, formula class $\Phi^{\mathcal{I}}_{\mathbf{Y},\mathbf{T}}\subseteq\Phi_{\mathbf{Y},\mathbf{T}}$
abstractly as assigning to each latent projection $\mathcal{G}^p\in\mathcal{A}(\mathbf{O})$ and disjoint node sets $\mathbf{Y},\mathbf{T}\subseteq\mathbf{O}$,
a set of formulas
$\mathcal{I}(\mathcal{G}^p,\mathbf{Y},\mathbf{T})\subseteq\Phi^{\mathcal{I}}_{\mathbf{Y},\mathbf{T}}$.
The procedure is \emph{sound} (for identification)
if each $\phi \in \mathcal{I}(\mathcal{G}^p,\mathbf{Y},\mathbf{T})$
is identifying for $\mathbf{Y}\mid\doop(\mathbf{T})$.
The procedure is \emph{complete} (for identification)
if it returns at least one identifying formula 
whenever $\mathbf{Y}\mid\doop(\mathbf{T})$ is identifiable in $\mathcal{G}^p$.
The procedure is
\emph{exhaustively complete relative to $\Phi^\mathcal{I}_{\mathbf{Y},\mathbf{T}}$}
if it is complete for finding all identifying formulas in $\Phi^\mathcal{I}_{\mathbf{Y},\mathbf{T}}$,
that is,
for every $\phi \in \Phi^\mathcal{I}_{\mathbf{Y},\mathbf{T}}$,
if $\phi$ is identifying for $\mathbf{Y}\mid\doop(\mathbf{T})$,
then $\phi\in\mathcal{I}(\mathcal{G}^p,\mathbf{Y},\mathbf{T})$.

\section{The verification problem}\label{sec:guarantees}

Identification asks for some identifying observational formula for a target $\mathbf{Y}\mid\doop(\mathbf{T})$.
Verification is the complementary decision problem:
given a graph, a target, and a proposed answer,
decide whether that answer is correct.
We formalize this decision task as follows.
\begin{definition}
[Verifier]
\label{def:verifier}
    A \emph{verifier} $\mathcal{V}$ is a decision procedure, that is,
    a Boolean-valued algorithm that halts on every valid input,
    computing the following map.
    It
    takes as input a latent projection $\mathcal{G}^p\in\mathcal{A}(\mathbf{O})$, disjoint node sets $\mathbf{Y},\mathbf{T}\subseteq\mathbf{O}$, and either an observational formula $\phi\in\Phi_{\mathbf{Y},\mathbf{T}}$ or the symbol $\texttt{none}$.
    Let $\mathcal{G}$ be any latent-variable causal directed acyclic graph whose latent projection onto $\mathbf O$ is $\mathcal G^p$.
    The verifier returns $\texttt{true}$ if and only if one of the following holds:
    \begin{enumerate}
        \item $\phi\in\Phi_{\mathbf{Y},\mathbf{T}}$ and $\phi$ is identifying for $\mathbf{Y}\mid\doop(\mathbf{T})$ in $\mathcal{G}$; or
        \item $\phi = \texttt{none}$ and $\mathbf{Y}\mid\doop(\mathbf{T})$ is not identifiable in $\mathcal{G}$.
    \end{enumerate}
    Otherwise,
    $\mathcal{V}$ returns 
    $\texttt{false}$.
\end{definition}

The verification task is not solved by identification alone.
A sound and complete identification procedure
need only return some identifying formula when one exists,
and therefore need not decide whether an arbitrary observational formula is identifying.
This applies, for instance,
to the ID algorithm,
which halts on every valid input and returns
an identifying formula when the target is identifiable,
and reports non-identifiability otherwise \citep{Shpitser2006}.
One can also obtain some more, but not necessarily all, 
identifying formulas using the approach of \citet{Yvernes2026}.
\zcref[S]{fig:id-different-formulas} illustrates
why this is not enough to decide whether an arbitrary observational formula is identifying.

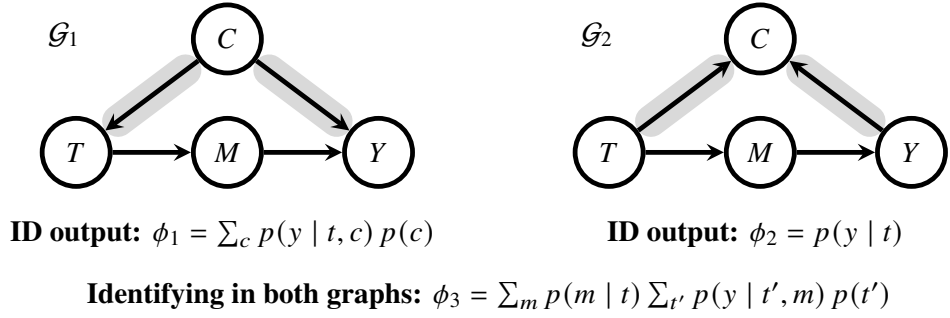
\begin{figure}[t]
\centering
\figwidth=29pc
\begin{minipage}[t]{0.45\textwidth}
\centering
\begin{tikzpicture}[>=stealth, line width=1.8pt]
    \node[draw=none, anchor=north west, inner sep=0pt] at (-0.5,1.9) {$\mathcal{G}_1$};

    \node (T) at (0,0) {$T$};
    \node (M) at (2,0) {$M$};
    \node (Y) at (4,0) {$Y$};
    \node (C) at (2,1.5) {$C$};

    \draw[->] (T) -- (M);
    \draw[->] (M) -- (Y);
    \draw[->, highlight] (C) -- (T);
    \draw[->, highlight] (C) -- (Y);
\end{tikzpicture}
\par\medskip
\textbf{ID output:} $\phi_1=\sum_{c} p(y\mid t,c)\,p(c)$
\par\medskip
\end{minipage}\hspace{0.5em}%
\begin{minipage}[t]{0.45\textwidth}
\centering
\begin{tikzpicture}[>=stealth, line width=1.8pt]
    \node[draw=none, anchor=north west, inner sep=0pt] at (-0.5,1.9) {$\mathcal{G}_2$};

    \node (T) at (0,0) {$T$};
    \node (M) at (2,0) {$M$};
    \node (Y) at (4,0) {$Y$};
    \node (C) at (2,1.5) {$C$};

    \draw[->] (T) -- (M);
    \draw[->] (M) -- (Y);
    \draw[->, highlight] (T) -- (C);
    \draw[->, highlight] (Y) -- (C);
\end{tikzpicture}
\par\medskip
\textbf{ID output:} $\phi_2=p(y \mid t)$
\par\medskip
\end{minipage}
\par\medskip
\textbf{Identifying in both graphs:}
$\phi_3=\sum_{m} p(m \mid t)\sum_{t^\prime} p(y \mid t^\prime, m)\,p(t^\prime)$
\caption{
The ID algorithm returns
one identifying formula for $\mathbf{Y}\mid\doop(\mathbf{T})$
when one exists; here,
the adjustment formula $\phi_1$ in $\mathcal{G}_1$
and the conditional density $\phi_2$ in $\mathcal{G}_2$.
Each formula is specific to its graph and fails in the other: 
$\mathcal{V}(\mathcal{G}_1, Y, T, \phi_2)=\texttt{false}$ and $\mathcal{V}(\mathcal{G}_2, Y, T, \phi_1)=\texttt{false}$.
Conversely, non-return by ID
is not evidence of a formula being non-identifying:
the front-door formula $\phi_3$
is identifying in both graphs,
but is not the formula returned by ID in either.
Thus, completeness for identification is an existence guarantee,
insufficient to enable verification of arbitrary formulas.
}
\label{fig:id-different-formulas}
\end{figure}

While an identification procedure $\mathcal{I}$ that is sound and exhaustively complete
relative to a formula class $\Phi^\mathcal{I}_{\mathbf{Y},\mathbf{T}}$, such as adjustment formulas \ifsupplementary
    (see Section~S3 of the Supplementary Material),
\else
    (see \zcref[S]{app:adjustment}),
\fi
may be repurposed for verification
by checking whether a formula belongs to $\mathcal{I}(\mathcal{G}^p,\mathbf{Y},\mathbf{T})$,
it can at best verify formulas within that class,
not arbitrary formulas in $\Phi_{\mathbf{Y},\mathbf{T}}$.
One might therefore try to make the class as broad as possible,
but this does not remove the difficulty
and instead shifts it to deciding membership in
$\mathcal{I}(\mathcal{G}^p,\mathbf{Y},\mathbf{T})$.

The natural broad route is proof search:
try to verify a candidate observational formula
by searching for a derivation of that formula
in the
do-calculus proof system (\citealp[][Section~4]{Pearl1995}; see also \zcref[S]{app:do-calculus}).
Soundness of do-calculus guarantees that observational formulas derived by a sequence of derivation steps starting from the target are identifying for the target,
and completeness guarantees that
some identifying formula is derivable whenever the target is identifiable
\citep{Shpitser2006,Huang2006}.
Since our observational formulas are written in the same symbolic density language used in do-calculus derivations,
this is a meaningful route.
However, it turns verification into a derivability problem.
We next formalize this problem and show that direct proof search via fair enumeration of derivations only semi-decides it:
derivable formulas are eventually accepted,
with the derivation serving as a certificate,
whereas non-derivable formulas need not lead to termination.

\subsection{The limits of do-calculus proof search for verification}\label{subsec:limits-do-calc}

Let $\Sigma$ be a finite alphabet containing the symbols needed to write
the density expressions, interventions, algebraic operations, and marginalizations considered below.
Let $\Sigma^*$ denote the
set of finite strings over $\Sigma$
and let
$\mathcal{L}^{\mathrm{do\text{-}calc}}\subseteq\Sigma^*$
be the language of well-formed
symbolic density expressions used in do-calculus 
and probability-algebra derivations.
Since we fix this syntax throughout,
membership in the language $\mathcal{L}^{\mathrm{do\text{-}calc}}$
and in the subclass
$\Phi^{\mathrm{do\text{-}calc}}_{\mathbf{Y},\mathbf{T}}\subseteq\mathcal{L}^{\mathrm{do\text{-}calc}}$
of strings representing observational formulas
for $\mathbf{Y}$ under intervention on $\mathbf{T}$
is decidable, in the sense that there is an algorithm
that halts on every input and correctly determines membership.
For simplicity, we take $\mathcal{R}$ to be a finite set of
sound rule schemas for symbolic density expressions,
including the do-calculus and probability-algebra schemas;
finiteness simplifies the enumeration argument below,
but effective enumerability of rules
and decidable rule applicability would suffice.

\begin{definition}
[Derivation]
\label{def:derivation}
    A \emph{derivation} of length $m\in\mathbb{N}$
    from a starting query expression $\xi^0=Q_{\mathbf{Y},\mathbf{T}}$
    is a finite sequence $d_m=(j_1,\dots,j_m)$ that determines
    a sequence of expressions 
    $\xi^1,\ldots,\xi^m\in\mathcal{L}^\mathrm{do\text{-}calc}$
    by a computable transition function.
    For all $k\in\{1,\dots,m\}$, $j_k:a_k\rightsquigarrow b_k$
    is a local rewrite step, where $a_k$ is an occurrence of a 
    probability-kernel sub-expression of $\xi^{k-1}$, and $b_k$
    is the expression obtained from $a_k$ by one application of a 
    rule schema in $\mathcal{R}$.
    The next expression $\xi^k$ is obtained from $\xi^{k-1}$ by
    replacing the selected occurrence of $a_k$ with $b_k$.
    The derivation is valid relative to $\mathcal{G}^p$
    if all steps are well-formed and all graph-dependent side conditions
    of the applied rules
    hold in $\mathcal{G}^p$.
\end{definition}

Given a candidate observational formula $\phi\in\Phi_{\mathbf{Y},\mathbf{T}}^\mathrm{do\text{-}calc}$, a
derivation $d_m$ derives $\phi$ from $Q_{\mathbf{Y},\mathbf{T}}$ in $\mathcal{G}^p$ if the expression obtained
after applying all its rewrite steps, $\xi^m$, is 
syntactically equal to $\phi$; the following example illustrates this.

\begin{example}[Front-door formula derivation]
Consider the graph $\mathcal{G}_1$ in
\zcref[S]{fig:id-different-formulas} and the query
$\xi^0=Q_{Y,T}=p(y\mid\doop(t))$.
Using probability manipulations and the do-calculus rules
stated in \zcref[S]{app:do-calculus},
a possible derivation of the front-door formula $\phi_3$ 
is the finite sequence
$d_7=(j_1,\dots,j_7)$, where
\begin{alignat*}{2}
j_1:\quad
&\stepstrut
p(y\mid\doop(t))
\rightsquigarrow
\sum_m p(y\mid\doop(t),m)p(m\mid\doop(t)),
&\qquad&
\text{marginalizing over $M$};\\
j_2:\quad
&\stepstrut
p(m\mid\doop(t))
\rightsquigarrow
p(m\mid t),
&&
\text{by Rule 2 of do-calculus};\\
j_3:\quad
&\stepstrut
p(y\mid\doop(t),m)
\rightsquigarrow
p(y\mid\doop(t,m)),
&&
\text{by Rule 2 of do-calculus};\\
j_4:\quad
&\stepstrut
p(y\mid\doop(t,m))
\rightsquigarrow
p(y\mid\doop(m)),
&&
\text{by Rule 3 of do-calculus};\\
j_5:\quad
&\stepstrut
p(y\mid\doop(m))
\rightsquigarrow
\sum_{t'}p(y\mid\doop(m),t')p(t'\mid\doop(m)),
&&
\text{marginalizing over $T$};\\
j_6:\quad
&\stepstrut
p(t'\mid\doop(m))
\rightsquigarrow
p(t'),
&&
\text{by Rule 3 of do-calculus};\\
j_7:\quad
&\stepstrut
p(y\mid\doop(m),t')
\rightsquigarrow
p(y\mid m,t'),
&&
\text{by Rule 2 of do-calculus}.
\end{alignat*}
Applying these local rewrites successively to $\xi^0$ yields $\phi_3$.
Thus, we say that this particular derivation 
$d_7$ derives $\phi_3$ from $Q_{Y,T}$ in $\mathcal{G}_1$.
\end{example}

Let $D_m$ be the set of all length-$m$ derivations, 
and consider
$Q_{\mathbf{Y},\mathbf{T}}\coloneqq p(\mathbf{y}\mid \doop(\mathbf{t}))$
as starting query expression. Let $\textsc{Check}(G^p,Q_{\mathbf{Y},\mathbf{T}},\phi,d_m)$
be an algorithm taking as input 
the graph $\mathcal{G}^p$ and query $Q_{\mathbf{Y},\mathbf{T}}$,
an observational formula $\phi\in\Phi_{\mathbf{Y},\mathbf{T}}^\mathrm{do\text{-}calc}$,
and a derivation $d_m\in D_m$;
$\textsc{Check}$ returns $\texttt{true}$ if 
$d_m$
is a valid derivation relative to $\mathcal{G}^p$
that derives $\phi$,
and $\texttt{false}$ otherwise.

We consider the set of observational formulas
that are derivable
from $Q_{\mathbf{Y},\mathbf{T}}$
in finitely many steps,
and are hence identifying formulas, as
\[
\mathcal{I}^{\mathrm{do\text{-}calc}}(\mathcal{G}^p,\mathbf{Y},\mathbf{T})
=
\{\phi\in\Phi_{\mathbf{Y},\mathbf{T}}^{\mathrm{do\text{-}calc}}\mid\exists m\in\mathbb{N}, d_m\in D_m:\textsc{Check}(\mathcal{G}^p,Q_{\mathbf{Y},\mathbf{T}},\phi,d_m)=\texttt{true}\},
\]
which formally defines a language over $\Sigma$
as well as
a sound and complete identification procedure.
This do-calculus-based identification procedure
may even be exhaustively complete relative to
$\Phi^\mathrm{do\text{-}calc}_{\mathbf{Y},\mathbf{T}}$.
Nevertheless,
this does not by itself provide a verifier for arbitrary candidate formulas,
because membership in the derivable set $\mathcal{I}^\mathrm{do\text{-}calc}(\mathcal{G}^p,\mathbf{Y},\mathbf{T})$ is only semi-decidable.

Recall that a language $L\subseteq \Sigma^*$ is \emph{semi-decidable} 
if there is an algorithm
that halts and accepts on inputs in $L$,
while it may run forever on inputs outside $L$.

\begin{restatable}[Semi-decidability of derivation search]
    {theorem}{semidecidability}
\label{thm:semi-decidability}
$\mathcal{I}^{\mathrm{do\text{-}calc}}
(\mathcal{G}^p,\mathbf{Y},\mathbf{T})$
is semi-decidable.
\end{restatable}
We provide a proof in \zcref[S]{app:proof-semi-decidability}.
It constructs an algorithm that fairly enumerates all candidate derivations
and
halts once it finds one that derives a given $\phi\in\Phi^\mathrm{do\text{-}calc}_{\mathbf{Y},\mathbf{T}}$; 
if none exists,
the search continues forever.

\begin{example}
[Non-terminating derivation search]
\label{ex:semi-decidability}

    Consider the graph $\mathcal{G}_2$ in
    \zcref[S]{fig:id-different-formulas} and the query
    $Q_{Y,T}=p(y\mid\doop(t))$.
    To see that the adjustment formula $\phi_1$ 
    is not identifying for $Y\mid\doop(T)$ in $\mathcal{G}_2$,
    we provide a counterexample for which the formula and the 
    target disagree.
    Consider the Gaussian density $p$ that factorizes 
    according to $\mathcal{G}_2$ with 
    $T\sim\mathcal{N}(0,1)$, $M\mid T=t\sim\mathcal{N}(t,1)$,
    $Y\mid M=m\sim\mathcal{N}(m,1)$, and
    $C\mid T=t, Y=y\sim\mathcal{N}(t+y,1)$.
    Under $\doop(T=1)$, the interventional distribution is 
    $Y\mid\doop(T=1)\sim\mathcal{N}(1,2)$.
    On the other hand, $C\sim\mathcal{N}(0,7)$ and, for fixed $c$, 
    $Y|T=1,C=c\sim{N}(2c/3-1/3,2/3)$, and hence
    \[
    \llbracket\phi_1\rrbracket(p,1)(y)=
    \mathcal{N}\left(y;-\frac{1}{3},\frac{34}{9}\right)
    \ne\mathcal{N}(y;1,2)=p_{Y\mid\doop(T=1)}(y).
    \]

    We explain what happens if one tries to verify $\phi_1$ 
    by derivation search.
    For all $m\in\mathbb{N}$ and all $d_m\in D_m$, 
    $\textsc{Check}(\mathcal{G}_2,Q_{Y,T},\phi_1,d_m)=\texttt{false}$,
    because otherwise, by soundness of $\mathcal{R}$, $\phi_1$ 
    would be identifying, contradicting the counterexample above.
    However, this does not allow the procedure to halt and reject, 
    because the derivation space has no finite bound.
    Indeed, even from the target expression, one can insert
    probabilistic identities leaving the represented kernel unchanged.
    For example, we may rewrite
    \begin{equation*}
        p(y\mid\doop(t))\rightsquigarrow p(y\mid\doop(t))\sum_c p(c\mid y,\doop(t))\rightsquigarrow p(y\mid\doop(t)),
    \end{equation*}
    because $\sum_c p(c\mid y,\doop(t))=1$. This is a valid two-step
    derivation loop. For all $k\in\mathbb{N}$, one may insert this
    loop $k$ times before applying any other rewrite, which gives
    arbitrarily long valid candidate derivations. 
    After checking finitely many candidate derivations, the procedure
    has ruled out only finitely many candidates, but there remain longer
    derivations that have not yet been checked.
\end{example}

\zcref[S]{thm:semi-decidability} does not rule out the existence
of a terminating decision procedure for derivability.
Such a procedure would exist, for example, if one could
compute a finite bound $B$, possibly depending on the input, such that,
whenever $\phi$ is derivable, it has a valid derivation of length
at most $B$. One could then enumerate all derivations up to length $B$,
accept if one of them derives $\phi$, and reject otherwise.
Alternatively, a terminating procedure might search the derivation 
space while somehow ignoring redundant detours such as the one in 
\zcref[S]{ex:semi-decidability}.
We are not aware of such a terminating procedure for the 
derivability problem considered here, and this problem may even be undecidable, meaning that
no algorithm can correctly decide all instances while halting on 
every input. While we do not prove such a result, this
impossibility is in line with known undecidability
results for closely related probabilistic and causal
reasoning problems \citep{Ibeling2025}.
The verification task therefore remains open as a separate problem,
despite a rich identification literature.

\section{A falsification-based verifier}\label{sec:hiprof}

\subsection{Falsification procedure}\label{subsec:falsification-procedure}
Our strategy to verification is based on \emph{falsification}:
instead of attempting to prove that \zcref[S]{eq:identifying-formula} holds 
for all densities factorizing according to the graph, we search for a counterexample violating it.
This yields a
verification procedure for parametric submodels (\zcref[S]{thm:as-correct-verifier}).

We consider a parametric family 
$\{p_\theta\}_{\theta\in\Theta}\coloneqq\{\{p_{\theta_V}(v\mid \mathbf v_{\operatorname{pa}(V)})\}_{\theta_V\in\Theta_V}\}_{V\in\mathbf{V}}$ 
of densities factorizing according to $\mathcal{G}$, 
where, for all $V\in\mathbf{V}$, $\Theta_V\subseteq\mathbb{R}^{d_V}$, 
and $\Theta\coloneqq\prod_{V\in\mathbf{V}}\Theta_V\subseteq\mathbb{R}^d$. 
For all $V\in\mathbf{V}$, let $\pi_V$ be a distribution on $\Theta_V$
that is absolutely continuous with respect to the Lebesgue measure,
and let $\pi \coloneqq \bigotimes_{V\in\mathbf{V}}\pi_V$ be the joint distribution
on $\Theta$.
We write $\mathcal{P}_{\Theta}(\mathcal{X}_{\mathbf{V}})\subseteq\mathcal{P}_{\mathcal{G}}(\mathcal{X}_{\mathbf{V}})$ 
for the parametric submodel induced by this family.
Then, a falsifier is defined as follows.

\begin{definition}
[Falsifier]
\label{def:falsifier}
    A \emph{falsifier} $\mathcal{F}$ is a decision procedure, that is,
    a Boolean-valued algorithm that halts on every valid input,
    computing the following map.
    It takes as input a latent projection $\mathcal{G}^p\in\mathcal{A}(\mathbf{O})$,
    disjoint node sets $\mathbf{Y},\mathbf{T}\subseteq\mathbf{O}$, 
    and either an observational formula $\phi\in\Phi_{\mathbf{Y},\mathbf{T}}$ 
    or the symbol $\texttt{none}$.
    Let $\mathcal{G}$ be any latent-variable causal directed acyclic graph 
    whose latent projection onto $\mathbf O$ is $\mathcal G^p$,
    and consider a parametric family $\{p_\theta\}_{\theta\in\Theta}$
    of densities factorizing according to $\mathcal{G}$.
    Let $\theta_1,\dots,\theta_K$, with $K\in\mathbb{N}$, be 
    independent draws from $\pi$.
    Then, conditioned on $\{p_\theta\}_{\theta\in\Theta}$ and on the 
    realized parameter values,
    $\mathcal{F}$ returns $\texttt{true}$
    if and only if one of the following holds:
    \begin{enumerate}
        \item \label{itm:ae-equality-check}
        $\phi\in\Phi_{\mathbf{Y},\mathbf{T}}$, and,
        for all $i\in\{1,\dots,K\}$
        and all $\mathbf{t}\in\mathcal{X}_{\mathbf{T}}$,
        \begin{equation*}
            \llbracket \phi \rrbracket (p_{\theta_i,\mathbf{O}},\mathbf{t})=
            p_{\theta_i,\mathbf{Y}\mid\doop(\mathbf{T}=\mathbf{t})}
            \quad \mu_{\mathbf{Y}}\text{-almost everywhere; or}
        \end{equation*}
        \item $\phi=\texttt{none}$ and $\mathbf{Y}\mid\doop(\mathbf{T})$ is not identifiable in $\mathcal{G}$.
        \label{itm:id-check}
    \end{enumerate}
    Otherwise, $\mathcal{F}$ returns $\texttt{false}$.
\end{definition}
We say that a falsifier is an almost-surely correct verifier relative to
$\mathcal{P}_{\Theta}(\mathcal{X}_{\mathbf{V}})$
if, with probability one over the sampled parameters, 
it returns $\texttt{true}$ exactly when either 
$\phi\in\Phi_{\mathbf{Y},\mathbf{T}}$ is identifying for 
$\mathbf{Y}\mid\doop(\mathbf{T})$ in $\mathcal{G}$ relative to 
$\mathcal{P}_{\Theta}(\mathcal{X}_{\mathbf{V}})$,
or $\phi=\texttt{none}$ and $\mathbf{Y}\mid\doop(\mathbf T)$ 
is not identifiable in $\mathcal{G}$.
When $\phi$ is $\texttt{none}$ (Case~\ref{itm:id-check}), 
we use the ID algorithm to check identifiability,
which has been shown to be sound and complete (for identification) \citep{Shpitser2006}.
Case~\ref{itm:ae-equality-check}, instead, involves two nontrivial tasks:
(a) deciding whether two densities agree
$\mu_{\mathbf{Y}}$-almost everywhere, and
(b) checking this equality for all intervention values
$\mathbf t\in\mathcal X_{\mathbf T}$.
Both tasks can be difficult for general parametric families,
but in our implementation we use the canonical directed acyclic
graph associated with $\mathcal{G}^p$, where each bidirected edge $a \leftrightarrow b$ is represented by 
an additional variable $u$ with $a \leftarrow u \rightarrow b$ 
(an alternative implementation avoids specifying the latent structure;
we discuss its trade-offs in \zcref[S]{app:linear-sems-margs}),
and use a linear Gaussian parametrization, which makes
the above problems tractable. To obtain the
interventional density, we remove all incoming edges into the treatment
variables and set the treatment variables to their intervened values.
In this linear Gaussian setting, all relevant densities are Gaussian
(see the closure result in \zcref[S]{app:gaussian-closure}). 
Therefore, task~(a) reduces to comparing mean vectors and
covariance matrices. Moreover, by the same closure result,
admissible formula outputs have mean affine in $\mathbf t$ and covariance
independent of $\mathbf t$, and for fixed parameters,  
task~(b) reduces to comparing the covariance matrices
and comparing the mean functions at $|\mathbf T|+1$ affinely
independent intervention values.

\begin{example}
[Falsification in a Gaussian model]
\label{ex:falsification-example}
Consider the acyclic directed mixed graph
$\mathcal G^p: T \gets C \leftrightarrow Y$
and
the query $Q_{Y,T}=p(y\mid\doop(t))$.
For parametric falsification,
consider the canonical directed acyclic graph
$\mathcal{G}: T \gets C \gets L \to Y$
with the centred linear Gaussian model:
$L \sim \mathcal{N}(0,\sigma_L^2)$,
$C\mid L=l \sim \mathcal{N}(\lambda_{LC}l, \sigma_C^2)$,
$T\mid C=c \sim \mathcal{N}(\lambda_{CT}c, \sigma_T^2)$,
and
$Y\mid L=l \sim \mathcal{N}(\lambda_{LY}l, \sigma_Y^2)$,
where $\theta=(\lambda_{CT},\lambda_{LC},\lambda_{LY},\sigma_T^2,\sigma_C^2,\sigma_L^2,\sigma_Y^2)$
collects the parameters
of the induced joint.
The centering is only for exposition:
intercepts leave the covariance calculations unchanged
and add affine terms to the means.
For compactness, write
\[
\begin{aligned}
q       &\coloneqq \lambda_{LC}^2\sigma_L^2+\sigma_C^2,
&
v_T     &\coloneqq \lambda_{CT}^2q+\sigma_T^2,
&
v_Y     &\coloneqq \lambda_{LY}^2\sigma_L^2+\sigma_Y^2,
\\
s_{TC}  &\coloneqq \lambda_{CT}q,
&
s_{YT}  &\coloneqq \lambda_{CT}\lambda_{LC}\lambda_{LY}\sigma_L^2,
&
s_{YC}  &\coloneqq \lambda_{LC}\lambda_{LY}\sigma_L^2,
&
\Delta  &\coloneqq v_Tq-s_{TC}^2.
\end{aligned}
\]

The induced joint, observed joint,
and interventional distribution are Gaussian, and
by the closure result in \zcref[S]{app:gaussian-closure},
each admissible formula returns a Gaussian density.
Thus checking $\mu_Y$-almost everywhere equality 
reduces to comparing mean and covariance parameters.

Under the intervention $\doop(T=t)$, truncating the factor 
for childless $T$ yields
$p_{\theta,Y\mid\doop(T=t)}(y)=\mathcal{N}(y;\ 0,\ v_Y)$.
We now compare this target with the outputs of
two candidate observational formulas:
$\phi_1 \coloneqq p(y\mid t)$ and $\phi_2 \coloneqq \sum_c p(y\mid t,c)p(c)$.
The expressions below are obtained mechanically from the 
observed joint by Gaussian marginalization, conditioning, and
kernel composition. We intentionally leave the resulting
rational expressions unsimplified to reflect the form
manipulated by the implementation.
In principle, equality could be checked symbolically by
reducing the resulting rational polynomial identities
(showing it is decidable),
but in practice this becomes computationally expensive beyond
toy examples;
the two formulas below already illustrate
how quickly the expressions grow.

Consider $\phi_1$ first. Marginalizing the observed joint to 
$(T,Y)$ and conditioning on $T=t$ yields
\[
\llbracket\phi_1\rrbracket(p_\theta,t)(y) = 
\mathcal{N}\left(y;\
\frac{s_{YT}}{v_T}t
,\
v_Y - \frac{s_{YT}^2}{v_T}
\right).\]
These parameters generally differ from the target parameters.
Thus, a single sampled parameter value at which either the mean 
or the variance differs is enough to falsify $\phi_1$.
However, on the lower-dimensional subset $\{\theta:\lambda_{CT}=0\}$,
we have $s_{YT}=0$, and hence $\phi_1$ agrees with the target
for every $t$, even though it is not identifying in the graph.

Consider now $\phi_2$.
Starting from the observed joint, conditioning gives the first factor,
marginalization gives the second,
and marginalizing their product over $c$ gives
\begin{align}
&\llbracket\phi_2\rrbracket(p_\theta,t)(y) \\
&\quad= 
\mathcal{N}\left(y;\
\frac{s_{YT}q-s_{YC}s_{TC}}{\Delta}t
,\
v_Y
-
\frac{
s_{YT}^2q
-
2s_{YT}s_{YC}s_{TC}
+
s_{YC}^2v_T
}{
\Delta
}
+
\left(
\frac{-s_{YT}s_{TC}+s_{YC}v_T}{\Delta}
\right)^2q
\right).
\end{align}
Symbolic simplification reduces the displayed mean to $0$ and the 
displayed variance to $v_Y$, so $\phi_2$ agrees with the target.

In the falsification procedure, we instead compare the induced 
Gaussian parameters at sampled parameter values $\theta_1,\dots,\theta_K$.
The equality is required for all intervention values $t$, but in
the linear Gaussian setting the means are affine in $t$ and the 
variances are independent of $t$ (\zcref[S]{app:gaussian-closure}).
Since $|T|=1$ here, it is enough to compare the mean functions at two distinct intervention values $t_0, t_1$.
A disagreement at any sampled parameter value and 
intervention value falsifies the formula.
Agreement at finitely many 
sampled parameter values
does not prove identification: a non-identifying formula can
agree with the target accidentally on special parameter values,
as $\phi_1$ does when $\lambda_{CT}=0$.
Below, we show that such accidents
form measure-zero sets under the sampling distribution.

\end{example}

\subsection{Almost-sure correct verifier}
When $\phi\in\Phi_{\mathbf{Y},\mathbf{T}}$ is identifying, 
\zcref[S]{eq:identifying-formula} holds for all densities factorizing according to the graph whenever $\phi$ is well-defined;
hence, a counterexample found by the falsifier certifies that $\phi$ is non-identifying. 
If the falsifier does not find a counterexample,
two cases remain possible: (i) 
$\phi$ is non-identifying relative to 
$\mathcal{P}_{\Theta}(\mathcal{X}_{\mathbf{V}})$,
but the sampled parameter values happen to lie in a set
on which the formula agrees with the target 
(for instance,
in \zcref[S]{ex:falsification-example}, no witnessing counterexample is found
for $\phi_1$ if all sampled parameter values lie in
$\{\theta:\lambda_{CT}=0\}$);
(ii) $\phi$ is identifying relative to 
$\mathcal{P}_{\Theta}(\mathcal{X}_{\mathbf{V}})$.
We address the first case by restricting our focus on conditional exponential families and show that, under regularity assumptions, (i) happens only on a nowhere dense measure-zero subset of the
parameter space (\zcref[S]{prop:measure-zero}).
\begin{definition}
[Conditional exponential-family parametrization]
\label{def:exp_family}
    A parametric family $\{p_\theta\}_{\theta\in\Theta}$ of densities 
    factorizing according to $\mathcal{G}$ is said to admit a 
    \emph{conditional exponential-family parametrization}
    if, for all $V\in\mathbf{V}$ and $\theta_V\in\Theta_V$,
    \begin{equation*}
        p_{\theta_V}(v \mid  \mathbf v_{\operatorname{pa}(V)}) = b_V(v,  \mathbf v_{\operatorname{pa}(V)})\,e^{\eta_V(\theta_V)^\top s_V(v,  \mathbf v_{\operatorname{pa}(V)}) - A_V\left(\eta_V(\theta_V),  \mathbf v_{\operatorname{pa}(V)}\right)},
    \end{equation*}
    where $b_V:\mathcal{X}_V\times\mathcal{X}_{\operatorname{pa}(V)}\to[0,\infty)$ is a non-negative function, $\eta_V:\Theta_V\to\mathbb{R}^{k_V}$ is the natural parameter, and $s_V:\mathcal{X}_V\times\mathcal{X}_{\operatorname{pa}(V)}\to\mathbb{R}^{k_V}$ is the sufficient statistic, with $k_V\in\mathbb{N}$. For all $V\in\mathbf{V}$, $\theta_V\in\Theta_V$ and $\mathbf v_{\operatorname{pa}(V)}\in\mathcal{X}_{\operatorname{pa}(V)}$, $A_V(\eta_V(\theta_V),  \mathbf v_{\operatorname{pa}(V)})<\infty$, where $A_V(\eta_V(\theta_V),  \mathbf v_{\operatorname{pa}(V)})\coloneqq\log{\int_{\mathcal{X}_V} b_V(v, \mathbf v_{\operatorname{pa}(V)})e^{\eta_V(\theta_V)^\top s_V(v,  \mathbf v_{\operatorname{pa}(V)})}\mathrm{d}\mu_V(v)}$
    is the normalizer.
\end{definition}

Given a conditional exponential-family parametrization $\{p_\theta\}_{\theta\in\Theta}$ of densities factorizing 
according to $\mathcal{G}$ and an observational formula $\phi$, 
we introduce the following assumptions, which are satisfied for the linear Gaussian submodel in our implementation.
\begin{assumption}
[Open and connected parameter space]
\label{ass:open-connected-parspace}
    For all $V \in \mathbf{V}$, $\Theta_V$ is open and connected.
\end{assumption}
\begin{assumption}
[Analyticity of the natural parameters]
\label{ass:eta-analytic}
    For all $V\in\mathbf{V}$, the map 
    $\theta_V\mapsto\eta_V(\theta_V)$ is analytic with open image $\eta_V(\Theta_V)$.
\end{assumption}
\begin{assumption}
[Regularity]
\label{ass:regular}
    For all $\mathbf t\in\mathcal{X}_{\mathbf{T}}$,
    $\{p_{\theta,\mathbf{X}\mid\doop(\mathbf{T}=\mathbf{t})}(\mathbf{x})\}_{\theta\in\Theta}$,
    where $\mathbf{X}=\mathbf{V}\setminus\mathbf T$, form a regular exponential family.
\end{assumption}

Regularity implies that, for all $\theta\in\Theta$ and all $\mathbf t\in\mathcal{X}_{\mathbf{T}}$, 
$p_{\theta,\mathbf{X}\mid\doop(\mathbf{T}=\mathbf{t})}(\mathbf{x})\propto\widetilde{b}_{\mathbf{t}}(\mathbf{x})e^{\widetilde{\eta}_{\mathbf{t}}(\theta)^\top\widetilde{s}_{\mathbf{t}}(\mathbf{x})}$, 
where $\widetilde{b}_{\mathbf{t}}:\mathcal{X}_{\mathbf{X}}\to[0,\infty)$, 
$\widetilde{\eta}_{\mathbf{t}}:\Theta\to\mathbb{R}^{k_{\mathbf{t}}}$, and 
$\widetilde{s}_{\mathbf{t}}:\mathcal{X}_{\mathbf{X}}\to\mathbb{R}^{k_{\mathbf{t}}}$, with ${k_{\mathbf{t}}}\in\mathbb{N}$, 
are such that the parametrization is minimal, full, 
and $\widetilde{\eta}_{\mathbf{t}}(\Theta)$ is open \citep[][p. 116]{BarndorffNielsen2014}. 
Here, minimality refers to affine independence of the components of
$\widetilde{\eta}_{\mathbf{t}}$ and 
$\mu_{\mathbf{x}}$-almost sure affine independence of the components of $\widetilde{s}_{\mathbf{t}}$, 
while fullness means that 
$\widetilde{\eta}_{\mathbf{t}}(\Theta)=\{\gamma_{\mathbf{t}}\in\mathbb{R}^{k_{\mathbf{t}}}:\int\widetilde{b}_{\mathbf{t}}(\mathbf{x})e^{\gamma_{\mathbf{t}}^\top\widetilde{s}_{\mathbf{t}}(\mathbf{x})}\mathrm{d}\mu_{\mathbf{x}}(\mathbf{x})<\infty\}$.
This assumption 
holds, for example, for discrete and Gaussian but not arbitrary conditional exponential-family parametrizations \citep{Boeken2026}.

\begin{assumption}
[Analyticity of the candidate formula]
\label{ass:phi-analytic}
    There exists a measurable set $E\subseteq\mathcal{X}_{\mathbf Y}$ of full measure such that, for all $\mathbf y\in E$ and all $\mathbf t\in\mathcal{X}_{\mathbf{T}}$, the map $\theta\mapsto\llbracket \phi \rrbracket (p_{\theta,\mathbf{O}},\mathbf{t})(\mathbf y)$ is analytic on $\Theta$.
\end{assumption}
Under \zcref[S]{ass:open-connected-parspace, ass:eta-analytic, ass:regular}, the observational marginal densities that form the base terms of our grammar are analytic in the parameters
\citep[][Theorem~8]{Boeken2026}. In the non-degenerate linear Gaussian case considered in our implementation, the same holds for the conditional densities. Indeed, the mean and covariance parameters of the marginal and conditional Gaussian base terms are analytic in $\theta$: the Gaussian mean and covariance are analytic functions of the natural parameters, while marginalization and conditioning involve only block extraction and analytic operations. By the closure result in \zcref[S]{app:gaussian-closure}, every formula satisfying our grammar therefore yields a Gaussian density whose mean and covariance are obtained from those of the base terms through finitely many operations that preserve analyticity. Since a Gaussian density is obtained from
its mean and covariance through compositions of real-analytic functions, it is analytic in these parameters; therefore, every admissible
formula satisfies \zcref[S]{ass:phi-analytic}.

\begin{restatable}
[Generic failure of non-identifying formulas]
{proposition}{measurezero}
\label{prop:measure-zero}
    Define the set of parameters
    for which the target interventional density agrees with the observational formula for all $\mathbf{t}\in\mathcal{X}_{\mathbf{T}}$ as
    \[
    S\coloneqq\{\theta\in\Theta:\text{for all }\mathbf t\in\mathcal{X}_{\mathbf{T}},\,\llbracket \phi \rrbracket (p_{\theta,\mathbf{O}},\mathbf{t})=p_{\theta,\mathbf{Y}\mid\doop(\mathbf{T}=\mathbf{t})}\,\,\mu_{\mathbf Y}\text{-almost everywhere}\}.
    \]
    Let $E$ be the full-measure set from \zcref[S]{ass:phi-analytic}.
    Under \zcref[S]{ass:open-connected-parspace, ass:eta-analytic, ass:regular, ass:phi-analytic}, 
    if there exists $\theta^*\in\Theta$ and $\mathbf t^*\in\mathcal{X}_{\mathbf{T}}$ such that   
    $\mu_{\mathbf{Y}}(\{
    \mathbf{y}\in E:
    \llbracket \phi \rrbracket (p_{\theta^*,\mathbf{O}},\mathbf{t}^*)(\mathbf y)
    \ne 
    p_{\theta^*,\mathbf{Y}\mid\doop(\mathbf{T}=\mathbf{t}^*)}(\mathbf y)
    \})>0$,
    then $S$ has Lebesgue measure zero.
\end{restatable}
We prove the result in \zcref[S]{app:proof-measure-zero}
by establishing analyticity of the interventional density,
and combining it with the analyticity of the candidate formula
ensured by our grammar. Their difference is therefore 
analytic, and the measure-zero conclusion follows from 
the identity theorem \citep{Mityagin2020} for real-analytic
functions: the zero set of a non-zero analytic function
on an open connected domain has Lebesgue measure zero.

The falsifier therefore never rejects an identifying 
formula and, for a non-identifying formula, returns a 
counterexample with probability one relative to the chosen 
parametric family. Combining these properties with the 
soundness and completeness of the ID algorithm yields the 
following result, proved in \zcref[S]{app:proof-as-correct-verifier}.

\begin{restatable}
[Almost-surely correct verifier]
{theorem}{ascorrectverifier}
\label{thm:as-correct-verifier}
    For conditional exponential-family parametrizations, under \zcref[S]{ass:open-connected-parspace, ass:eta-analytic, ass:regular, ass:phi-analytic}, the falsifier in \zcref[S]{def:falsifier} induces an almost-surely correct verifier relative to $\mathcal{P}_{\Theta}(\mathcal{X}_{\mathbf{V}})$.
\end{restatable}
In light of \zcref[S]{thm:as-correct-verifier},
$K=1$ in 
\zcref[S]{def:falsifier} suffices for almost-sure correctness 
relative to $\mathcal{P}_{\Theta}(\mathcal{X}_{\mathbf{V}})$
under exact evaluation and absolutely continuous parameter sampling. 
These conditions are not met by pseudo-random floating-point implementations,
and comparisons up to a fixed tolerance do not inherit the
same guarantee (see \zcref[S]{app:exact} for an example).
In the linear Gaussian case, however, 
\zcref[S]{eq:identifying-formula} is a rational function of the mean 
and covariance parameters; after clearing denominators, verification 
therefore reduces to polynomial identity testing, which is decidable by
exact symbolic procedures \citep[][Chapter~4]{Shpilka2010}. 
Since symbolic procedures are often computationally expensive, 
and efficient deterministic procedures are not available in general, 
falsification remains justified as a practical randomized procedure: 
we sample parameters from a large finite integer set and evaluate the 
polynomial using exact arithmetic.
This addresses the gap
between the idealized assumptions and actual computational implementation:
it removes floating-point error
and 
the polynomial identity testing bound
controls the finite-sampling 
probability of falsely accepting a non-identifying formula. See \zcref[S]{app:exact} for details.

\section{Verification for front-door gateways}\label{sec:gateway}

Consider identification relative to the class of front-door formulas
$\Phi_{\mathbf{Y},\mathbf{T}}^\mathrm{fd}\coloneqq\{\phi_{\mathbf{Z}}^{\mathrm{fd}}\coloneqq\sum_{\mathbf{z}}p(\mathbf{z}\mid \mathbf t)\sum_{\mathbf t'}p(\mathbf y\mid \mathbf t', \mathbf{z})\,p(\mathbf t'):\mathbf{Z}\subseteq\mathbf{O}\setminus(\mathbf{T}\cup\mathbf{Y})\}$. 
The front-door criterion \citep[][Section~3.2]{Pearl1995} gives graphical conditions
on a candidate set $\mathbf{Z}$
sufficient for the corresponding front-door formula
to be identifying
for $\mathbf{Y}\mid\doop(\mathbf{T})$ 
(see \zcref[S]{app:frontdoor}).
These conditions have been described as overly restrictive
\citep[][Section~3.3.2]{Pearl2009}. 
Indeed, the front-door criterion
is sound but not exhaustively complete relative to
$\Phi_{\mathbf{Y},\mathbf{T}}^\mathrm{fd}$: 
there exist sets $\mathbf{Z}$ for which $\phi_{\mathbf{Z}}^{\mathrm{fd}}$
is identifying for $\mathbf{Y}\mid\doop(\mathbf{T})$
even though $\mathbf{Z}$ does not satisfy the criterion. 
Consider for instance \zcref[S]{fig:front-door}, which shows 
an acyclic directed mixed graph 
with unobserved confounding between $T$ and $Y$. 
Here, neither $\{M,A\}$ nor $\{M,C\}$ satisfies the
front-door criterion, since there is a back-door path from these 
sets to $Y$ through $B$ that is not blocked by $T$. 
Nevertheless, our falsifier certifies the front-door
formula for either set as identifying relative to the
parametric submodel it considers. In this example, we can
in fact prove that the formulas are identifying
in the full non-parametric model; see \zcref[S]{app:proof-front-door}.
A successful do-calculus derivation would also certify 
such formulas, but derivation search is not a general 
halting verifier of candidate formulas (\zcref[S]{thm:semi-decidability}).
Therefore, the front-door criterion is not exhaustively complete
relative to $\Phi_{\mathbf{Y},\mathbf{T}}^\mathrm{fd}$:
it can fail to certify front-door formulas that are 
identifying nonetheless.

\begin{figure}[t]
    \centering
    \figwidth=24pc
    \begin{tikzpicture}[>=stealth, line width=1.8pt]
        \node (T) at (0,0) {$T$};
        \node (M) at (2,0) {$M$};
        \node (Y) at (4,0) {$Y$};
        \node (C) at (6,0) {$C$};
        \node (B) at (8,0) {$B$};

        \node (A) at (6,1.5) {$A$};

        \draw[->] (T) -- (M);
        \draw[->] (M) -- (Y);
        \draw[->] (C) -- (Y);
        \draw[->] (B) -- (C);

        \draw[->] (B) -- (A);
        \draw[->] (A) -- (Y);

        \draw[<->, bend left=65] (T) to (Y);
    \end{tikzpicture}
    \caption{Acyclic directed mixed graph with unobserved confounding between $T$ 
    and $Y$. Although $\{M,A\}$ and $\{M,C\}$ do not satisfy the
    front-door criterion, the front-door adjustment formula can 
    still be certified as identifying. 
    Example based on \citet[][Figure~1]{Wienobst2024}.}
    \label{fig:front-door}
\end{figure}
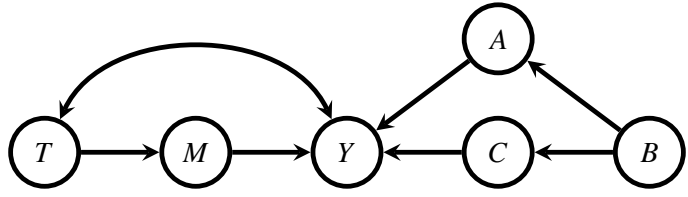

Verification allows us to develop the \emph{gateway test}, 
a procedure that is sound and exhaustively complete relative to $\Phi_{\mathbf{Y},\mathbf{T}}^\mathrm{fd}$.
The gateway test enumerates all candidate sets $\mathbf{Z}\subseteq\mathbf{O}\setminus(\mathbf{T}\cup\mathbf{Y})$,
constructs
the corresponding front-door formula $\phi_{\mathbf{Z}}^\mathrm{fd}\in\Phi_{\mathbf{Y},\mathbf{T}}^\mathrm{fd}$, 
and applies the verifier to it. 
A candidate set $\mathbf{Z}$ passes the test if and only if
$\phi_{\mathbf{Z}}^\mathrm{fd}$ is verified as identifying.
Thus, with an exact verifier, the gateway test returns
exactly those candidate sets
whose front-door formulas are identifying.
In practice, if one uses the falsifier from \zcref[S]{sec:hiprof},
the procedure is almost-surely sound and exhaustively complete
relative to $\Phi_{\mathbf{Y},\mathbf{T}}^\mathrm{fd}$
for the conditional exponential family chosen in the falsification routine.
The same idea applies to any finite class of candidate formulas.
When a formula class is indexed by candidate sets,
a graphical criterion
can replace the verifier in the inner loop,
but only if it is exhaustively complete relative to that class
(see \zcref[S]{app:sound-and-complete}):
for each candidate set,
it must decide whether the associated formula is identifying,
rather than merely 
guarantee that some valid set is found whenever one exists
(the latter suffices for completeness for identification).

\section{Discussion and outlook}\label{sec:outlook}

Our strategy to verification restricts the model class, which 
makes verification decidable and tractable. However, restricting the model 
class is not a general solution, since other parametric families
can involve non-algebraic expressions leading to an undecidable
decision problem \citep{Richardson1968}.

The falsification strategy comes with guarantees that are relative 
to the parametric submodel $\mathcal{P}_{\Theta}(\mathcal{X}_{\mathbf V})$
used by the falsifier 
(\zcref[S]{thm:as-correct-verifier}): if a candidate formula 
is certified by the falsifier, it is identifying 
relative to 
$\mathcal{P}_\Theta(\mathcal{X}_{\mathbf V})$ almost surely,
but it need not be valid in the full
non-parametric graphical model
$\mathcal{P}_{\mathcal{G}}(\mathcal{X}_{\mathbf V})$.
Characterizing when correctness relative to 
$\mathcal{P}_\Theta(\mathcal{X}_{\mathbf V})$ transfers to 
correctness in $\mathcal{P}_{\mathcal{G}}(\mathcal{X}_{\mathbf V})$
remains open.
Possible directions include using increasingly rich parametric
families or relaxing \zcref[S]{def:obsformula} to focus on mean effects
rather than interventional distributions.
Finally, verification can be extended
beyond unconditional 
interventional targets 
by allowing both sides of \zcref[S]{eq:identifying-formula} to be functionals in observational and interventional densities,
including (in)equality constraints implied by the graph \citep{Sachs2026},
as well as beyond our formula grammar in \zcref[S]{def:obsformula},
for example,
to include formulas
in linear instrumental-variable models
identifying mean effects instead of
full interventional distributions.

\bibliography{paper-ref}
\newpage

\appendix
\clearpage
\addcontentsline{toc}{section}{Appendix}
\listofappendices
\vskip 0.75in
\renewcommand{\theequation}{\thesection.\arabic{equation}}
\ifarxiv
    \appsection{Preliminaries}\label{app:preliminaries}
\else
    \section{Preliminaries}\label{app:preliminaries}
\fi

\subsection{Graphical preliminaries}

A graph $\mathcal{G}=(\mathbf{V},\mathbf{E})$ consists of a finite non-empty
node set $\mathbf{V}$, whose nodes represent random variables, 
and an edge set $\mathbf{E}$. 
Nodes joined by an edge are called \emph{adjacent},
and an edge joining two nodes is \emph{incident} to those nodes.
If the edge set contains only 
ordered pairs of distinct vertices, that is, if 
$\mathbf{E}\subseteq \{(V_i,V_j)\in \mathbf{V}\times \mathbf{V}:V_i\neq V_j\}$,
then $\mathcal{G}$ is a \emph{directed graph}; for 
$(V_i,V_j)\in \mathbf{E}$, we write $V_i\to V_j$. 
Directed graphs contain at most one edge between any pair
of distinct nodes.
A \emph{directed mixed graph} may instead contain two types of edges:
directed ($\to$) and bidirected ($\leftrightarrow$), with at most 
one edge of each type between any pair of distinct nodes.

A \emph{walk} $w$ in $\mathcal{G}$ is a sequence of nodes
$V_1,\dots,V_j\in\mathbf{V}$ and a corresponding sequence of edges
$e_1,\dots,e_{j-1}$ such that, for all $i\in\{1,\dots,j-1\}$,
$e_i$ is an edge between $V_i$ and $V_{i+1}$.
The first node $V_1$ and the last node $V_j$ are called the
\emph{endpoints} of $w$. A \emph{path} is a walk whose nodes are
distinct. A walk, or path, 
from a set $\mathbf{X}\subseteq\mathbf{V}$ to a disjoint set
$\mathbf{Y}\subseteq\mathbf{V}$ is a walk, or path, 
from some $X\in\mathbf{X}$ to some $Y\in\mathbf{Y}$. 
Such a walk, or path, is \emph{proper} if only its
first node belongs to $\mathbf{X}$.
A \emph{back-door path} from a set $\mathbf{X}\subseteq\mathbf{V}$
to a disjoint set $\mathbf{Y}\subseteq\mathbf{V}$ is a proper path
whose first edge has an arrowhead into a node in $\mathbf{X}$.
A \emph{directed walk}, or path, from $V_1$ to $V_j$ is a walk, or path, whose edges are
all directed and point from $V_1$ towards $V_j$. 
A \emph{directed cycle} is a directed walk
$V_1,\dots,V_j,V_1$ with $V_1,\dots,V_j$ distinct.
A directed graph without directed cycles is a \emph{directed acyclic
graph}, whereas a directed mixed graph without directed cycles is an
\emph{acyclic directed mixed graph}. If $V_i\to V_j$, then $V_i$ is a
\emph{parent} of $V_j$. For all $V\in\mathbf{V}$, we write
$\operatorname{pa}(V)$ for the set of parents of $V$ in $\mathcal{G}$.

In the latent-variable setting, we write 
$\mathbf{V}=\mathbf{O}\sqcup\mathbf{L}$, where $\mathbf{O}$
and $\mathbf{L}$ denote the observed and latent variables, 
respectively. We represent a latent-variable directed acyclic
graph $\mathcal{G}$ by an acyclic directed mixed graph 
$\mathcal{G}^p$ obtained via \emph{latent projection} 
\citep{Verma1990,Richardson2003}, defined as follows.
\begin{definition}[\citealp{Richardson2023}, Definition~A.2]
    Let $\mathcal{G}$ be a latent-variable directed acyclic graph
    with node set $\mathbf{V}=\mathbf{O}\sqcup\mathbf{L}$,
    where nodes in $\mathbf{O}$ are observed and nodes in 
    $\mathbf{L}$ are unobserved. The \emph{latent projection} of 
    $\mathcal{G}$ onto $\mathbf{O}$ is an acyclic directed mixed
    graph $\mathcal{G}^p$ where, for all distinct 
    $O_i,O_j\in\mathbf{O}$:
    \begin{enumerate}[label=(\roman*)]
        \item $\mathcal{G}^p$ contains a directed edge 
        $O_i\to O_j$ if $\mathcal{G}$ has a directed path from 
        $O_i$ to $O_j$ whose non-endpoint nodes, if any, all 
        belong to $\mathbf{L}$.
        \item $\mathcal{G}^p$ contains a bidirected edge 
        $O_i\leftrightarrow O_j$ if $\mathcal{G}$ has a path 
        between $O_i$ and $O_j$ such that 
        $O_i\leftarrow L_1\leftarrow\cdots\leftarrow L_h\to\cdots\to L_m\to O_j$, with $L_k\in\mathbf{L}$
        for all $k\in\{1,\dots,m\}$.
    \end{enumerate}
\end{definition}

Given a directed acyclic graph $\mathcal{G}$ over nodes 
$\mathbf{V}$, let $w$
be a walk with node sequence 
$V_1,\dots,V_j\in\mathbf{V}$. A non-endpoint node $V_i$, with
$i\in\{2,\dots,j-1\}$, is a \emph{collider} on $w$ if 
the two edges on $w$ incident to $V_i$ both have arrowheads
into $V_i$, that is, if $w$ contains the subwalk
$V_{i-1}\rightarrow V_i \leftarrow V_{i+1}$.
A non-endpoint node on $w$ that is not a collider is called a
\emph{non-collider}.
These notions lead to the definition of $d$-separation.
\begin{definition}
[$d$-connection and separation]
    Given a directed acyclic graph $\mathcal{G}$ over nodes 
    $\mathbf{V}$, a walk $w$ in $\mathcal{G}$ is
    \emph{open} given a set 
    $\mathbf{S}\subseteq\mathbf{V}\setminus\{V_1,V_m\}$ if
    \begin{enumerate}
        \item every collider on $w$ belongs to $\mathbf{S}$; and
        \item every non-collider on $w$ does not belong to $\mathbf{S}$.
    \end{enumerate}
    A walk that is not open given $\mathbf{S}$ is \emph{blocked}
    given $\mathbf{S}$. For all disjoint subsets $\mathbf{A}, \mathbf{B}, \mathbf{S}\subseteq \mathbf{V}$, 
    we say that $\mathbf{A}$ 
    and $\mathbf{B}$ are \emph{$d$-connected} by $\mathbf{S}$
    if there exist $A\in\mathbf{A}$, $B\in\mathbf{B}$, and
    an open walk from $A$ to $B$ given $\mathbf{S}$. Otherwise,
    $\mathbf{A}$ and $\mathbf{B}$ are \emph{$d$-separated}
    by $\mathbf{S}$.
\end{definition}
This definition of $d$-separation is equivalent to
the path-based definition, such as those of
\citet[][Definition~1.2.3]{Pearl2009} and \citet[][Definition~6.1]{Peters2017}.

\subsection{Causal preliminaries}\label{app:causal-preliminaries}
Consider a latent-variable directed acyclic graph 
$\mathcal{G}$ over $\mathbf{V}=\mathbf{O}\sqcup\mathbf{L}$ and 
suppose that all directed edges in $\mathcal{G}$ represent causal
relationships. Under this interpretation, for all 
$V_1,V_2\in\mathbf{V}$, a directed edge $V_1\to V_2$ indicates 
that $V_1$ is a direct cause of of $V_2$; a directed path 
$V_1\to\cdots\to V_2$ indicates that $V_1$ is a cause of $V_2$; 
and a bidirected edge $V_1\leftrightarrow V_2$ indicates that 
$V_1$ and $V_2$ share an unobserved common cause.
For all $V\in\mathbf{V}$, let $\mathcal{X}_V$ be the sample space 
of $V$, and let $\mu_V$ be a $\sigma$-finite measure on $\mathcal{X}_V$.
Throughout, we consider densities in $\mathcal{P}_{\mathcal{G}}(\mathcal{X}_{\mathbf{V}})$,
that is, densities that factorize according to $\mathcal{G}$ 
(see \zcref[S]{sec:causal-model-identification}). 
Every $q\in\mathcal{P}_{\mathcal{G}}(\mathcal{X}_{\mathbf{V}})$ 
satisfies the \emph{global Markov property} with respect
to $\mathcal G$: for all disjoint
$\mathbf A,\mathbf B,\mathbf S\subseteq\mathbf V$, if
$\mathbf A$ and $\mathbf B$ are $d$-separated by $\mathbf S$ in
$\mathcal G$, then $\mathbf{A}$ and $\mathbf{B}$ are
conditionally independent given $\mathbf{S}$ under $q$.

Consider the intervention node set $\mathbf T\subseteq \mathbf O$,
and intervention values $\mathbf t\in \mathcal X_{\mathbf T}$. 
Under the intervention $\operatorname{do}(\mathbf T=\mathbf t)$ 
(or shorthand $\operatorname{do}(\mathbf{t})$), 
the variables in $\mathbf T$ are set to $\mathbf t$. 
Let $\mathbf X\coloneqq\mathbf O\setminus \mathbf T$. 
The resulting interventional density is given by the
truncated factorization formula \citep[][Section~1.3]{Pearl2009}
(interventions may change the reference measure;
we nevertheless adopt the standard notation for simplicity):
\begin{align*}
    p_{\mathbf{X}\mid\doop(\mathbf{T}=\mathbf{t})}(\mathbf{x})
    =\int_{\mathcal{X}_{\mathbf{L}}}p(\mathbf{x},\boldsymbol{\ell})\,\mathrm{d}\mu_{\mathbf{L}}(\boldsymbol{\ell}),
\end{align*}
where 
\begin{align*}
    p(\mathbf x,\boldsymbol \ell)=\prod_{V\in \mathbf V\setminus \mathbf T} p(v\mid  \mathbf v_{\operatorname{pa}(V)})\big|_{\,\mathbf T=\mathbf t},
\end{align*}
that is, the conditional densities corresponding to the
intervened variables are removed from the factorization and
the remaining factors are evaluated after substituting $\mathbf T=\mathbf t$. 
The above expression, also known as the
g-formula \citep{Robins1996}, or the manipulated density 
\citep{Spirtes2000}, uses the factorization over the 
full latent-variable causal directed acyclic graph;
as only the observational marginal $p_{\mathbf{O}}$
is observed, the resulting interventional density need
not be identifiable from it. Identification concerns
precisely when the interventional can nevertheless be expressed as a 
functional of $p_{\mathbf{O}}$ satisfying constraints encoded by the latent projection $\mathcal{G}^p$ (see \zcref[S]{sec:causal-model-identification}).

\appsection{A typed, kernel-preserving grammar for observational formulas}\label{app:grammar}
We introduce a well-defined grammar for observational formulas,
designed to generate expressions that denote probability kernels
rather than arbitrary algebraic combinations of densities.

For disjoint sets
$\mathbf{A},\mathbf{B}\subseteq\mathbf{O}$, write
\[
E:\mathbf{A}\mid \mathbf B
\]
to mean that the expression \(E\) denotes a conditional density over \(\mathbf{A}\) given \(\mathbf{B}\),
or equivalently a density representation of a Markov kernel from
\(\mathcal X_{\mathbf B}\) to \(\mathcal X_{\mathbf A}\). This notation is a typing device:
it records the output variables of an expression and its conditioning
arguments.

\paragraph{Base terms.}
For disjoint sets $\mathbf{A},\mathbf{B}\subseteq\mathbf{O}$, observational marginals and conditionals
are admissible base terms:
\[
p(\mathbf a):\mathbf A\mid \varnothing,
\qquad
p(\mathbf a\mid \mathbf b):\mathbf A\mid \mathbf B.
\]
where \(p(\mathbf a\mid \mathbf b)\) is understood only on the part of $\mathcal{X}_{\mathbf B}$ on which $p(\mathbf b)>0$.
Although marginals and conditionals, where defined,
can be derived from the observed joint by
the rules below,
we take them as base terms to match the usual notation for identifying formulas.
Quotient notation such as $p(\mathbf a,\mathbf b)/p(\mathbf b)$
is understood only as shorthand for the corresponding conditional density,
not as an arbitrary division operation.

\paragraph{Well-formed products.}
Products are admissible only when they can be interpreted as
sequential compositions of kernels. Each factor must then introduce
a new set of output variables, so that every output variable
is assigned exactly to one factor. Moreover, every conditioning
variable must already be available when the factor is evaluated:
it must either be conditioned on by the product as a whole, or 
have been introduced by an earlier factor.

More formally, 
let \(\mathbf A_1,\ldots,\mathbf A_k,\mathbf C\subseteq \mathbf O\),
where \(\mathbf A_1,\dots,\mathbf A_k\) are pairwise disjoint and
disjoint from \(\mathbf C\). 
Suppose the factors can be ordered so that, for all $i\in\{1,\dots,k\}$,
\[
E_i : \mathbf A_i \mid \mathbf D_i,
\qquad
\mathbf D_i \subseteq \mathbf C\cup \mathbf A_1\cup\cdots\cup \mathbf A_{i-1}.
\]
A factor need not depend on all variables that are already available;
the requirement is only that it does not condition on variables that
have not yet been introduced.
Then, the product of the factors is 
\[
\prod_{i=1}^k E_i
:
\left(\bigsqcup_{i=1}^k\mathbf A_i\right)\mid \mathbf C.
\]
Thus, the product denotes a kernel over all variables introduced
by the factors, conditional on $\mathbf{C}$.

This rule includes ordinary chain-rule products and more general
chain-like constructions.
For example, $p(x)p(y\mid t,x)$ is admissible as a kernel over $(X,Y)$ conditional on $T$: the first factor 
introduces $X$, and the second factor introduces $Y$ while conditioning
only on $T$ and the already introduced variable $X$. The rule also
allows products such as $p(x)p(y)$, which
denotes a density over $(X,Y)$, although not necessarily the 
observational joint $p(x,y)$. By contrast, products such as $p(y)p(y)$
are not admissible, since the same output variable is introduced 
twice. Similarly, products such as $p(a\mid b)p(b\mid a)$ are not 
admissible either, since no ordering makes the conditioning variables 
available before the corresponding outputs are introduced.

\paragraph{Marginalization.}
Marginalization is admissible when it removes output variables
from a well-typed kernel. If an expression denotes a joint kernel
over disjoint sets of variables $\mathbf{A}$ and $\mathbf{B}$
conditional on $\mathbf{C}$, that is,
\[
E:\mathbf A,\mathbf B\mid \mathbf C,
\]
then, integrating out $\mathbf{B}$ leaves a kernel over $\mathbf{A}$,
conditional on the same variables $\mathbf{C}$:
\[
\sum_{\mathbf{b}} E(\mathbf a,\mathbf b\mid \mathbf c):\mathbf A\mid \mathbf C.
\]

\paragraph{Internal conditional division.}
The grammar allows variables to be moved from the output
side of an intermediate kernel to its conditioning side.
For pairwise disjoint sets $\mathbf{A}, \mathbf{B}, \mathbf{R}\subseteq\mathbf{O}$, if
\[
E:\mathbf A,\mathbf R\mid \mathbf B,
\]
then, for every fixed $\mathbf{b}$, 
$E(\mathbf{a},\mathbf{r}\mid\mathbf{b})$
denotes a joint density over $(\mathbf{A},\mathbf{R})$.
Let $\mathbf{M}\subseteq\mathbf{A}$. From the same kernel $E$,
form the conditional kernel of $\mathbf{R}$ given $(\mathbf{M},\mathbf{B})$,
denoted by $E(\mathbf{r}\mid\mathbf{m},\mathbf{b})$.
Where this conditional is defined, the grammar admits
\[
\frac{E(\mathbf{a},\mathbf r\mid \mathbf b)}{E(\mathbf r\mid \mathbf m,\mathbf b)}:\mathbf A\mid \mathbf B,\mathbf R.
\]
Thus, internal conditional division turns a kernel over
$(\mathbf{A},\mathbf{R})$ given $\mathbf{B}$ into a kernel over $\mathbf{A}$
given $(\mathbf{B},\mathbf{R})$. 
Ordinary conditioning of an intermediate kernel is the special
case with $\mathbf{M}=\emptyset$.

The denominator must be derived from the same kernel $E$ as the numerator.
A quotient with the same formal variables but with an independently 
supplied denominator is not generally kernel-preserving.
This is why the rule is an internal conditional division, rather than
an arbitrary algebraic quotient. Quotients are admissible
only when they are parsed as conditionals or as internal conditional
divisions, where kernel preservation follows from typing rule.

To see why this rule preserves kernel normalization, write 
$\mathbf{U}=\mathbf{A}\setminus\mathbf{M}$
such that $\mathbf{A} = (\mathbf{U},\mathbf{M})$
and recall that $E(\mathbf a,\mathbf r\mid \mathbf b)$
is, for each fixed $\mathbf b$, a joint density over $(\mathbf A,\mathbf R)$.
For every $\mathbf{b}$,
by the product rule,
\[
E(\mathbf{u},\mathbf{m},\mathbf{r}\mid\mathbf{b})
=
E(\mathbf u\mid \mathbf m,\mathbf r,\mathbf b)
E(\mathbf r\mid \mathbf m,\mathbf b)
E(\mathbf m\mid \mathbf b).
\]
After division by the internally derived conditional
\(E(\mathbf r\mid \mathbf m,\mathbf b)\),
the resulting expression is
\[
E(\mathbf u\mid \mathbf m,\mathbf r,\mathbf b)
E(\mathbf m\mid \mathbf b),
\]
which, for every \((\mathbf r,\mathbf b)\),
integrates to one over \((\mathbf U,\mathbf M)=\mathbf A\).

Internal conditional division is purely a typing rule: it is stated in terms of
kernels and does not refer to a graph. Graphical fixability is different:
it is a graph-based condition used to determine when a corresponding
fixing operation is valid \citep{Richardson2023}. When such a fixing
operation is written at the level of kernels, its division step is an
instance of the internal conditional division rule above: the denominator
is a conditional kernel derived from the same intermediate kernel as the
numerator. Thus, the expressions at the kernel level obtained from
fixing notation are covered by the grammar after expansion,
while fixability itself is not an additional typing rule.

\paragraph{Admissible formulas.}
An observational formula \(\phi\) for \(\mathbf Y\) under intervention
on \(\mathbf T\)
is an expression generated recursively by the preceding rules,
with no free variables other
than \(\mathbf y\) and \(\mathbf t\), and of type
$\mathbf Y\mid \mathbf T$,
which induces the partial functional
\[
\llbracket \phi \rrbracket:
\mathcal P(\mathcal X_{\mathbf O})
\times \mathcal X_{\mathbf T}
\rightharpoonup
\mathcal P(\mathcal X_{\mathbf Y}).
\]

This formalizes a convention often left implicit in identification
formulas:
products of density terms must be well-formed products,
sums or integrals must marginalize variables from a well-typed kernel,
and quotients must be conditionals or internal conditional divisions rather than arbitrary algebraic ratios.
For example,
$\sum_c p(y\mid t,c)p(c)$
is admissible and
has type \(Y\mid T\), 
whereas \(p(y)p(y)\) is not admissible merely as
an algebraic product of density symbols.

\subsection{Gaussian closure}
\label{app:gaussian-closure}
In the Gaussian case,
we restrict attention to non-degenerate Gaussian kernels,
that is, Gaussian kernels with strictly positive-definite
covariance matrices.
Marginalizations and well-formed products preserve
non-degeneracy, which ensures that all conditionals and
internal conditional divisions are defined everywhere.

Suppose the observational law over \(\mathbf O\) is multivariate 
Gaussian, $\mathbf O\sim \mathcal N(\mu,\Sigma)$, with $\Sigma\succ 0$. 
Then every admissible observational formula \(\phi\) of type 
\(\mathbf Y\mid \mathbf T\)
denotes, for all $\mathbf{t}\in\mathcal{X}_{\mathbf{T}}$,
a linear Gaussian density over \(\mathbf Y\): 
wherever
\(\llbracket\phi\rrbracket(p_{\mathbf{O}},\mathbf t)\) is defined, there exist
\(a_\phi\in\mathbb{R}^{|\mathbf{Y}|},B_\phi\in\mathbb{R}^{|\mathbf{Y}|\times|\mathbf{T}|},\Omega_\phi\in\mathbb{R}^{|\mathbf{Y}|\times|\mathbf{Y}|}\),
depending on \(p\) and \(\phi\) but not on
the intervention value \(\mathbf t\), such that
\[
\llbracket\phi\rrbracket(p_{\mathbf{O}},\mathbf t)(\mathbf y)
=
\mathcal N\left(\mathbf y;\ a_\phi+B_\phi \mathbf t,\ \Omega_\phi\right).
\]
Thus, the mean parameter is affine in \(\mathbf t\), while the covariance
parameter is independent of \(\mathbf t\).

The claim follows by induction over the grammar:
Gaussian marginals and conditionals are linear Gaussian kernels;
well-formed products of
compatible linear Gaussian kernels define joint linear Gaussian kernels;
marginalizing a proper Gaussian kernel again yields a Gaussian kernel
over the remaining variables;
and internal conditional division preserves linear Gaussianity
because the denominator is the conditional kernel computed from the same joint Gaussian kernel as the numerator, and the quotient
is therefore another conditional Gaussian kernel of that joint law.
Hence, every admissible formula remains linear Gaussian.
\ifarxiv
    \appsection{Verifying adjustment formulas with graphical criteria}\label{app:adjustment}
\else
    \section{Verifying adjustment formulas with graphical criteria}\label{app:adjustment}
\fi

Instead of asking whether the target interventional density is 
identifiable (by any observational formula),
one may ask whether it is identifiable by 
some member of a restricted class of formulas. 
Consider the class of adjustment formulas
\begin{equation}\label{eq:adjustment-class}
\Phi_{\mathbf{Y},\mathbf{T}}^{\mathrm{adj}}\coloneqq
\left\{
\phi_{\mathbf{Z}}^{\mathrm{adj}}\coloneqq
\sum_{\mathbf{z}} p(\mathbf y\mid \mathbf t,\mathbf z)\,p(\mathbf z)
:
\mathbf{Z}\subseteq\mathbf{O}\setminus(\mathbf T\cup\mathbf Y)
\right\}.
\end{equation}
Covariate adjustment asks whether there exists a set $\mathbf{Z}$
yielding an identifying formula for 
$\mathbf{Y}\mid\doop(\mathbf{T})$,
and graphical criteria for adjustment are conditions
on candidate sets $\mathbf{Z}$ that certify when
the corresponding adjustment formula 
$\phi_{\mathbf{Z}}^{\mathrm{adj}}$
is identifying for $\mathbf{Y}\mid\doop(\mathbf{T})$.
Different criteria give different guarantees.
The back-door criterion \citep[][Section~3.1]{Pearl1995}, 
for instance, is sound but not exhaustively complete 
relative to $\Phi_{\mathbf{Y},\mathbf{T}}^{\mathrm{adj}}$.
The adjustment criterion \citep[][Definition~5]{Shpitser2010} 
is instead sound and exhaustively complete 
relative to $\Phi_{\mathbf{Y},\mathbf{T}}^{\mathrm{adj}}$.
Both criteria are formulated for directed acyclic graphs, 
while others extend adjustment to richer graph classes.
In particular, the generalized adjustment criterion \citep[][Definition~4]{Perkovic2018} 
is sound and exhaustively complete 
relative to $\Phi_{\mathbf{Y},\mathbf{T}}^{\mathrm{adj}}$
in directed acyclic graphs, maximal ancestral graphs \citep{Richardson2002}, 
and their respective equivalence classes, 
thus allowing for the presence of unobserved variables. 
These criteria do not, however, apply to acyclic directed 
mixed graphs. 

In graph classes for which a sound and exhaustively complete 
graphical adjustment criterion is available, 
checking whether a candidate set $\mathbf{Z}$
satisfies the criterion is equivalent to verifying whether the corresponding adjustment formula $\phi_{\mathbf{Z}}^{\mathrm{adj}}$
is identifying for $\mathbf{Y}\mid\doop(\mathbf{T})$.

Such criteria can therefore be viewed as non-parametric graphical
shortcuts to verification,
but only in the graph class for which the criteria are valid.
When no such shortcuts are available, verifiers still apply.
For instance, even though no sound and exhaustively complete 
adjustment criterion is currently known for acyclic directed mixed graphs,
one can still enumerate the finite class 
$\Phi_{\mathbf{Y},\mathbf{T}}^{\mathrm{adj}}$
and verify each candidate adjustment formula directly
(see \zcref[S]{app:sound-and-complete}).

Moreover, adjustment criteria can be repurposed as verifiers
only for formulas in the restricted class 
$\Phi_{\mathbf{Y},\mathbf{T}}^{\mathrm{adj}}$.
When no candidate adjustment set satisfies
a sound and exhaustively complete graphical criterion
for the relevant graph class,
the target is not identifiable by any adjustment formula in
$\Phi_{\mathbf{Y},\mathbf{T}}^{\mathrm{adj}}$,
but it may still be identifiable.
In fully observed directed acyclic graphs, this
distinction is less visible for single-node interventions,
since adjustment suffices to determine identifiability in that special case.
Beyond this setting, however, the target need not be identifiable
by adjustment formulas, but it may still be identifiable
by other formulas in $\Phi_{\mathbf{Y},\mathbf{T}}$,
such as front-door formulas (\citealp[][Section~3.2]{Pearl1995}; see also 
\zcref[S]{sec:gateway} 
and \zcref[S]{app:frontdoor}), or by formulas returned 
by more general identification procedures,
such as the ID algorithm.
These more general routes rely on do-calculus
(\citealp[][Section~4]{Pearl1995}; see also \zcref[S]{app:do-calculus}),
but, as discussed in \zcref[S]{subsec:limits-do-calc}, 
proof search with do-calculus does not by itself yield a verifier for 
arbitrary formulas in $\Phi_{\mathbf{Y},\mathbf{T}}$.
Verifiers are therefore needed to check formulas that fall outside the 
scope of adjustment criteria.

\ifarxiv
    \appsection{Do-calculus}\label{app:do-calculus}
\else
    \section{Do-calculus}\label{app:do-calculus}
\fi

The interventional density can sometimes be identified even
when no set yields an identifying adjustment or front-door formula.
\citet[][Section~4]{Pearl1995} introduced a collection of three rules,
known as \emph{do-calculus}, that can be applied sequentially to 
rewrite interventional quantities
and may eventually lead to an expression in observational quantities
alone.

Let $\mathcal{G}$ denote the underlying latent-variable 
directed acyclic graph over observed nodes $\mathbf{O}$, and consider disjoint node sets $\mathbf{Y}, \mathbf{T}, 
\mathbf{Z}, \mathbf{W}\subseteq\mathbf{O}$. 
Do-calculus consists of the following rules:
\begin{enumerate}[label=(\roman*)]
    \item ``Insertion/deletion of observations'': $$p(\mathbf{y}\mid\operatorname{do}(\mathbf{t}),\mathbf{z},\mathbf{w})=p(\mathbf{y}\mid\operatorname{do}(\mathbf{t}),\mathbf{w}),$$ 
    if $\mathbf{Y}$ and $\mathbf{Z}$ are $d$-separated by $\mathbf{T}$, 
    $\mathbf{W}$ in the graph obtained by removing incoming edges into $\mathbf{T}$;
    \item ``Action/observation exchange'': $$p(\mathbf{y}\mid\operatorname{do}(\mathbf{t},\mathbf{z}),\mathbf{w})=p(\mathbf{y}\mid\operatorname{do}(\mathbf{t}),\mathbf{z},\mathbf{w}),$$
    if $\mathbf{Y}$ and $\mathbf{Z}$ are $d$-separated by $\mathbf{T}$, 
    $\mathbf{W}$ in the graph obtained by removing incoming edges
    into $\mathbf{T}$ and outgoing edges from $\mathbf{Z}$;
    \item ``Insertion/deletion of actions'':
    $$p(\mathbf{y}\mid\operatorname{do}(\mathbf{t}, \mathbf{z}), \mathbf{w})=p(\mathbf{y}\mid\operatorname{do}(\mathbf{t}),\mathbf{w}),$$
    if $\mathbf{Y}$ and $\mathbf{Z}$ are $d$-separated by 
    $\mathbf{T},\mathbf{W}$
    in the graph obtained 
    as follows: first remove all incoming
    edges into nodes in $\mathbf{T}$; then, in the
    resulting graph, identify those nodes of $\mathbf{Z}$ that are not
    ancestors of any node in $\mathbf{W}$;
    finally, remove all incoming edges into these nodes of $\mathbf{Z}$.
\end{enumerate}

These rules have been shown to be complete for identification \citep{Shpitser2006}.
For instance, do-calculus can be used to derive the front-door formula in \zcref[noname]{eq:front-door-formula}.
The ID algorithm \citep{Shpitser2006} allows one to derive one identifying formula, if one exists.
However,
as we discuss in \zcref[S]{sec:guarantees},
finding one formula with one derivation for it is not enough to determine for any proposed formula whether it can be derived using do-calculus.
\ifarxiv
    \appsection{Proofs}\label{app:proofs}
\else
    \section{Proofs}\label{app:proofs}
\fi

\subsection{Proof of \zcref[S]{thm:semi-decidability}}\label{app:proof-semi-decidability}

\semidecidability*

\begin{proof}
Showing that $\mathcal{I}^{\mathrm{do\text{-}calc}}(\mathcal{G}^p,\mathbf{Y},\mathbf{T})$ 
is semi-decidable is equivalent to showing that there exists 
an algorithm $M$ that semi-decides 
$\mathcal{I}^{\mathrm{do\text{-}calc}}(\mathcal{G}^p,\mathbf{Y},\mathbf{T})$, 
that is, an $M$ that halts and accepts for all inputs that are elements of 
$\mathcal{I}^{\mathrm{do\text{-}calc}}(\mathcal{G}^p,\mathbf{Y},\mathbf{T})$, 
but need not terminate otherwise.

The key point is to enumerate derivations fairly.
Indeed, for a fixed derivation length, there may be infinitely 
many derivations, so an enumeration that first exhausts all
candidates of length $1$, then all candidates of length $2$, and 
so on, may never reach longer derivations.
Instead, we enumerate pairs $(m,i)\in\mathbb{N}\times\mathbb{N}$,
where $m$ is the derivation length and $i$ is the index
within the enumeration of candidates of that length;
this ensures that $M$ can visit all derivations.

For all $m\in\mathbb{N}$, fix a computable enumeration 
$e_m:\mathbb{N}\to D_m$ of $D_m$, the set of derivations of length $m$.
Such an enumeration exists because each derivation has a finite 
encoding over $\Sigma$.
Let $f:\mathbb{N}\to\mathbb{N}\times\mathbb{N}$ be a computable
bijection that maps $n\in\mathbb{N}$ to a tuple 
$(m,i)\in\mathbb{N}\times\mathbb{N}$. 
The algorithm $M$ proceeds as follows.
For $n=0,1,2,\dots$, algorithm $M$ computes $f(n)=(m,i)$, sets
the derivation $d_{m,i}=e_m(i)$, and then runs
$\textsc{Check}(\mathcal{G}^p,Q_{\mathbf{Y},\mathbf{T}},\phi,d_{m,i})$
to check whether $d_{m,i}$ is a valid derivation relative to $\mathcal{G}^p$ that derives $\phi$ from $Q_{\mathbf{Y},\mathbf{T}}$. 
$\textsc{Check}$ halts for all inputs because it performs finitely
many operations on finite strings and evaluates only decidable 
graphical side-conditions (such as d-separation) on a finite graph.

If $\phi\in\mathcal{I}^{\mathrm{do\text{-}calc}}(\mathcal{G}^p,\mathbf{Y},\mathbf{T})$, 
then by definition there exists at least one derivation $d^*$ such 
that $\textsc{Check}(\mathcal{G}^p,Q_{\mathbf{Y},\mathbf{T}},\phi,d^*)=\texttt{true}$.
Therefore, there exists $n^*\in\mathbb{N}$ such that 
$f(n^*)=(m^*,i^*)$ and $d_{m^*,i^*}=e_{m^*}(i^*)=d^*$. 
When $M$ reaches $n^*$, it runs $\textsc{Check}$, accepts and halts.
If $\phi\notin\mathcal{I}^{\mathrm{do\text{-}calc}}(\mathcal{G}^p,\mathbf{Y},\mathbf{T})$, 
there exists no derivation for which $\textsc{Check}$
returns $\texttt{true}$. 
In particular, since $\textsc{Check}$ halts on every input, 
$M$ never encounters a derivation accepted by $\textsc{Check}$
and therefore fails to terminate.

This shows that $\mathcal{I}^{\mathrm{do\text{-}calc}}(\mathcal{G}^p,\mathbf{Y},\mathbf{T})$ semi-decides $\mathcal{L}$ and concludes the proof of 
\zcref[S]{thm:semi-decidability}.    
\end{proof}

\subsection{Proof of \zcref[S]{prop:measure-zero}}\label{app:proof-measure-zero}

\measurezero*

\begin{proof}
We follow the proof strategy used to show that, for exponential-family
parametrizations of densities factorizing according to the graph, 
under regularity assumptions, the set of parameter values for which the distribution is not faithful to the graph has 
Lebesgue measure zero \citep[][Theorems~7 and~8]{Boeken2026}.

Define the set of parameters for which the target interventional
density agrees with the observational formula on $E$
for all $\mathbf{t}\in\mathcal{X}_{\mathbf{T}}$ as
\begin{equation*}
    S_E\coloneqq\{\theta\in\Theta:\text{for all }\mathbf y\in E\text{ and all }\mathbf t\in\mathcal{X}_{\mathbf{T}},\,\llbracket \phi \rrbracket (p_{\theta,\mathbf{O}},\mathbf{t})(\mathbf y)=p_{\theta,\mathbf{Y}\mid\doop(\mathbf{T}=\mathbf{t})}(\mathbf y)\}.
\end{equation*}
We divide the proof in three steps:
\begin{enumerate}[label=(\roman*)]
    \item\label{itm:measure-zero} Using 
    \ifsupplementary
        Assumptions~1 and 4,
    \else
        \zcref[S]{ass:open-connected-parspace, ass:phi-analytic}, 
    \fi
    if, for all $\mathbf y\in E$ and all $\mathbf{t}\in\mathcal{X}_{\mathbf{T}}$, the real-valued map $\theta\mapsto p_{\theta,\mathbf{Y}\mid\doop(\mathbf{T}=\mathbf{t})}(\mathbf y)$ 
    is analytic on $\Theta$, 
    we show that $S_E$ has Lebesgue measure zero;

    \item\label{itm:measure-zero-ae} Using \ref{itm:measure-zero}, we show that $S$ has Lebesgue measure zero;
    
    \item\label{itm:analytic-density} Using 
    \ifsupplementary
        Assumptions~1--3,
    \else
        \zcref[S]{ass:open-connected-parspace, ass:eta-analytic, ass:regular}, 
    \fi
    we show that, for all $\mathbf y\in\mathcal{X}_{\mathbf Y}$ and all $\mathbf{t}\in\mathcal{X}_{\mathbf{T}}$, 
    the real-valued map $\theta\mapsto p_{\theta,\mathbf{Y}\mid\doop(\mathbf{T}=\mathbf{t})}(\mathbf y)$ is analytic on $\Theta$.
\end{enumerate}

Throughout, we use that sums, products, quotients
with non-zero denominator,
and compositions of analytic functions are analytic 
\citep[see][Chapter~3.2]{Conway1978}.

\paragraph{Proof of~\ref{itm:measure-zero}.}
    By the premises, 
    there exists $\theta^*\in\Theta$, $\mathbf{t}^*\in\mathcal{X}_{\mathbf{T}}$ and $\mathbf y^*\in E$ such that
    \begin{equation*}
        \llbracket \phi \rrbracket (p_{\theta^*,\mathbf{O}},\mathbf{t}^*)(\mathbf y^*)\ne p_{\theta^*,\mathbf{Y}\mid\doop(\mathbf{T}=\mathbf{t}^*)}(\mathbf y^*).
    \end{equation*}
    Define, for all $\theta\in\Theta$,
    \begin{equation*}
        g_{\mathbf{t}^*}(\theta)\coloneqq\llbracket \phi \rrbracket (p_{\theta,\mathbf{O}},\mathbf{t}^*)(\mathbf y^*) - p_{\theta,\mathbf{Y}\mid\doop(\mathbf{T}=\mathbf{t}^*)}(\mathbf y^*).
    \end{equation*}
    Then $g_{\mathbf{t}^*}(\theta^*)\ne0$ and hence $g_{\mathbf{t}^*}$ is not identically zero on $\Theta$. By 
    \ifsupplementary
        Assumption~4,
    \else
        \zcref[S]{ass:phi-analytic}, 
    \fi
    the real-valued map
    $\theta \mapsto \llbracket \phi \rrbracket (p_{\theta,\mathbf{O}},\mathbf{t}^*)(\mathbf y^*)$ 
    is analytic on $\Theta$.
    If the real-valued map
    $\theta \mapsto p_{\theta,\mathbf{Y}\mid\doop(\mathbf{T}=\mathbf{t}^*)}(\mathbf y^*)$ is analytic on $\Theta$, which we show in the proof of~\ref{itm:analytic-density},
    since the difference of analytic functions is analytic, 
    the real-valued map $\theta\mapsto g_{\mathbf{t}^*}(\theta)$ is analytic on $\Theta$ as well.
    
    Define $S_E^{\mathbf{t}^*,\mathbf y^*} \coloneqq \{\theta \in \Theta : g_{\mathbf{t}^*}(\theta) = 0\}$.
    Then, $S_E \subseteq S_E^{\mathbf{t}^*,\mathbf y^*}$.
    By the identity theorem \citep[][Proposition~1]{Mityagin2020},
    the zero set of a real analytic function that is not identically
    zero on an open and connected domain has 
    Lebesgue measure zero. Since, by
    \ifsupplementary
        Assumption~1,
    \else
        \zcref[S]{ass:open-connected-parspace}, 
    \fi
    for all $V\in\mathbf{V}$, $\Theta_V$ is open and connected, 
    $\Theta=\prod_{V\in\mathbf{V}}\Theta_V$ is open and connected
    as well. Therefore, $S_E^{\mathbf{t}^*,\mathbf y^*}$, and hence $S_E$, 
    has Lebesgue measure zero.
    This completes the proof of~\ref{itm:measure-zero}.
\paragraph{Proof of~\ref{itm:measure-zero-ae}.}
    Let $\theta^*\in\Theta$ and 
    $\mathbf{t}^*\in\mathcal{X}_{\mathbf{T}}$ be the witnesses 
    from the premise. 
    For all $\theta\in\Theta$ and $\mathbf{y}\in E$, define
    \[
    g_{\mathbf t^*}(\theta,\mathbf y)
    \coloneqq
    \llbracket \phi \rrbracket(p_{\theta,\mathbf O},\mathbf t^*)(\mathbf y)
    -
    p_{\theta,\mathbf{Y}\mid\doop(\mathbf{T}=\mathbf{t}^*)}(\mathbf y).
    \]
    Then the premise can be rewritten as 
    $\mu_{\mathbf Y}(\{\mathbf y\in E:g_{\mathbf t^*}(\theta^*,\mathbf y)\neq 0\})>0$.
    For all $\mathbf y\in E$, by \zcref[S]{ass:phi-analytic}, and if the 
    real-valued map
    $\theta \mapsto p_{\theta,\mathbf{Y}\mid\doop(\mathbf{T}=\mathbf{t}^*)}(\mathbf y)$ 
    is analytic on $\Theta$ (shown in the proof of~\ref{itm:analytic-density}), the real-valued map $\theta\mapsto g_{\mathbf t^*}(\theta,\mathbf y)$ is analytic on $\Theta$.

    Define
    \[
    S_{\mathbf t^*}
    \coloneqq
    \left\{
    \theta\in\Theta:
    \llbracket \phi \rrbracket(p_{\theta,\mathbf O},\mathbf t^*)(\mathbf y)
    =
    p_{\theta,\mathbf{Y}\mid\doop(\mathbf{T}=\mathbf{t}^*)}(\mathbf y)
    \,\,\mu_{\mathbf Y}\text{-almost everywhere}
    \right\}
    \]
    and, for all $\theta\in\Theta$,
    \[
    D_\theta\coloneqq
    \left\{
    \mathbf{y}\in\mathcal{X}_{\mathbf{Y}}:
    \llbracket \phi \rrbracket(p_{\theta,\mathbf O},\mathbf t^*)(\mathbf y)
    \neq
    p_{\theta,\mathbf{Y}\mid\doop(\mathbf{T}=\mathbf{t}^*)}(\mathbf y)
    \right\}
    =
    \left\{
    \mathbf{y}\in\mathcal{X}_{\mathbf{Y}}:
    g_{\mathbf{t}^*}(\theta,\mathbf{y})\ne0
    \right\}.
    \]
    By definition, for all $\theta\in\Theta$, 
    $\theta\in S_{\mathbf{t}^*}$ if 
    and only if $\mu_{\mathbf{Y}}(D_\theta)=0$. Since $E$ has
    full measure, $\mu_{\mathbf{Y}}(E^c)=0$. Moreover,
    for all $\theta\in\Theta$,
    $D_\theta=(D_\theta\cap E)\cup(D_\theta\cap E^c)$, and 
    $D_\theta\cap E^c\subseteq E^c$, which implies that
    $\mu_{\mathbf{Y}}(D_\theta)=0$ if and only if 
    $\mu_{\mathbf{Y}}(D_\theta\cap E)=0$.
    By definition of $g_{\mathbf{t}^*}$,
    $D_\theta\cap E=\{\mathbf{y}\in E:g_{\mathbf{t}^*}(\theta,\mathbf{y})\neq 0\}$. 
    Therefore, for all $\theta\in\Theta$, 
    \[
    \theta\in S_{\mathbf{t}^*}
    \quad\iff\quad
    \mu_{\mathbf Y}(\{\mathbf y\in E:g_{\mathbf t^*}(\theta,\mathbf y)\neq 0\})=0.
    \]

    Since $S$ requires the $\mu_{\mathbf{Y}}$-almost everywhere equality for all 
    $\mathbf t\in\mathcal X_{\mathbf T}$, $S\subseteq S_{\mathbf{t}^*}$.
    It is therefore enough to show that $S_{\mathbf{t}^*}$
    has Lebesgue measure zero, which implies that $S$ has Lebesgue measure zero.
    
    Suppose, by contradiction, that $S_{\mathbf{t}^*}$ has positive Lebesgue measure: $\lambda(S_{\mathbf{t}^*})>0$.
    For all $\theta\in S_{\mathbf{t}^*}$, by the equivalence established above,
    $\mu_{\mathbf Y}(\{\mathbf y\in E: g_{\mathbf t^*}(\theta,\mathbf y)\neq 0\})=0$.
    Therefore, by Tonelli's theorem:
    \begin{align*}
        0=\int_{S_{\mathbf{t}^*}}\mu_{\mathbf Y}\left(\left\{\mathbf y\in E:g_{\mathbf t^*}(\theta,\mathbf y)\neq 0\right\}\right)\mathrm d\lambda(\theta)
        &=
        \int_{S_{\mathbf{t}^*}}\int_E\mathbf 1_{\{g_{\mathbf t^*}(\theta,\mathbf y)\neq 0\}}\mathrm d\mu_{\mathbf Y}(\mathbf y)\mathrm d\lambda(\theta)\\
        &=
        \int_E\int_{S_{\mathbf{t}^*}}\mathbf 1_{\{g_{\mathbf t^*}(\theta,\mathbf y)\neq 0\}}\mathrm d\lambda(\theta)\mathrm d\mu_{\mathbf Y}(\mathbf y).
    \end{align*}
    Since the integrand is nonnegative, it follows that,
    for $\mu_{\mathbf{Y}}$-almost every $\mathbf{y}\in E$, $\lambda(\{\theta\in S_{\mathbf{t}^*}:g_{\mathbf t^*}(\theta,\mathbf y)\neq 0\})=0$. Therefore,
    \[
    \lambda(\{\theta\in S_{\mathbf{t}^*}:g_{\mathbf t^*}(\theta,\mathbf y)= 0\})=\lambda(S_{\mathbf{t}^*})>0,
    \]
    which means that, for $\mu_{\mathbf{Y}}$-almost every $\mathbf{y}\in E$,
    the zero set of the analytic map
    $\theta\mapsto g_{\mathbf t^*}(\theta,\mathbf y)$ has positive Lebesgue measure.
    By the identity theorem for real analytic functions, 
    a real analytic function on an open and connected domain whose zero 
    set has positive Lebesgue measure must be identically zero.
    Therefore, for $\mu_{\mathbf Y}$-almost every $\mathbf y\in E$,
    $g_{\mathbf t^*}(\cdot,\mathbf y)\equiv 0$ on $\Theta$, and, in particular, $g_{\mathbf t^*}(\theta^*,\mathbf y)=0$, $\mu_{\mathbf Y}$-almost everywhere on $E$.
    
    This contradicts the premise that 
    $\mu_{\mathbf Y}(\{\mathbf y\in E:g_{\mathbf t^*}(\theta^*,\mathbf y)\neq 0\})>0$.
    Hence the assumption that $S_{\mathbf t^*}$ has positive Lebesgue measure was false.
    It follows that $S_{\mathbf t^*}$ has Lebesgue measure zero. Since $S\subseteq S_{\mathbf t^*}$, we conclude that $S$ has Lebesgue measure zero.
    This completes the proof of~\ref{itm:measure-zero-ae}.
\paragraph{Proof of~\ref{itm:analytic-density}.}
    We now show that, under 
    \ifsupplementary
        Assumptions~1--3,
    \else
        \zcref[S]{ass:open-connected-parspace, ass:eta-analytic, ass:regular}, 
    \fi
    for all $\mathbf y\in\mathcal{X}_{\mathbf{Y}}$ and all $\mathbf{t}\in\mathcal{X}_{\mathbf{T}}$, 
    the real-valued map $\theta\mapsto p_{\theta,\mathbf{Y}\mid\doop(\mathbf{T}=\mathbf{t})}(\mathbf y)$ is analytic on $\Theta$, 
    which yields the premise of~\ref{itm:measure-zero} and,
    via \ref{itm:measure-zero} and \ref{itm:measure-zero-ae}, concludes the proof of
    \ifsupplementary
        Proposition~1.
    \else
        \zcref[S]{prop:measure-zero}. 
    \fi

    Fix an arbitrary $\mathbf{t}\in\mathcal{X}_{\mathbf{T}}$.
    By the truncated factorization formula 
    \citep[][Section~1.3]{Pearl2009}, 
    for all $\theta\in\Theta$ and  
    $\mathbf y\in\mathcal{X}_{\mathbf Y}$, 
    the interventional density is 
    $p_{\theta,\mathbf{Y}\mid\doop(\mathbf{T}=\mathbf{t})}(\mathbf y)=\int_{\mathcal{X}_{\mathbf{W}}}p_{\theta,(\mathbf{Y},\mathbf{W})\mid\doop(\mathbf{T}=\mathbf{t})}(\mathbf y,\mathbf{w})\mathrm{d}\mu_{\mathbf{W}}(\mathbf{w})$, with
    \begin{equation}\label{eq:truncated-factorization}
        p_{\theta,(\mathbf{Y},\mathbf{W})\mid\doop(\mathbf{T}=\mathbf{t})}(\mathbf y,\mathbf{w})=\prod_{V\in\mathbf{V}\setminus\mathbf T}p_{\theta_V}(v\mid \mathbf v_{\operatorname{pa}(V)})\big|_{\,\mathbf T=\mathbf t},
    \end{equation}
    where $\mathbf{W}\coloneqq\mathbf{V}\setminus(\mathbf T\sqcup\mathbf Y)$. For $V\in\mathbf{V}$, $\theta_V\in\Theta_V$, 
    $v\in\mathcal{X}_V$, and $\mathbf v_{\operatorname{pa}(V)\setminus\mathbf T}\in\mathcal{X}_{\operatorname{pa}(V)\setminus\mathbf T}$, 
    we define 
    \begin{equation*}
        b_V^*(v,\mathbf v_{\operatorname{pa}(V)\setminus\mathbf T},\mathbf t)\coloneqq b_V(v,\mathbf v_{\operatorname{pa}(V)})\big|_{\,\mathbf T=\mathbf t},
        \quad
        s_V^*(v,\mathbf v_{\operatorname{pa}(V)\setminus\mathbf T},\mathbf t)\coloneqq s_V(v,\mathbf v_{\operatorname{pa}(V)})\big|_{\,\mathbf T=\mathbf t},
    \end{equation*}
    and
    \begin{equation*}
        A_V^*(\eta_V(\theta_V),\mathbf v_{\operatorname{pa}(V)\setminus\mathbf T},\mathbf t)\coloneqq A_V(\eta_V(\theta_V),\mathbf v_{\operatorname{pa}(V)})\big|_{\,\mathbf T=\mathbf t},
    \end{equation*}
    which are obtained by substituting $\mathbf T=\mathbf t$ whenever $\mathbf T\subseteq\operatorname{pa}(V)$ and are
    such that
    \begin{align*}
        p_{\theta_V}(v\mid \mathbf v_{\operatorname{pa}(V)})\big|_{\,\mathbf T=\mathbf t}=b_V^*(v,\mathbf v_{\operatorname{pa}(V)\setminus\mathbf T},\mathbf t)e^{\eta_V(\theta_V)^\top s_V^*(v,\mathbf v_{\operatorname{pa}(V)\setminus\mathbf T},\mathbf t)-A_V^*(\eta_V(\theta_V),\mathbf v_{\operatorname{pa}(V)\setminus\mathbf T},\mathbf t)}.
    \end{align*}
    
    Recall that, by \zcref[S]{ass:regular}, if 
    $\mathbf{X}=\mathbf{V}\setminus\mathbf T$,
    for all $\theta\in\Theta$ and all $\mathbf t\in\mathcal{X}_{\mathbf{T}}$, 
    $p_{\theta,\mathbf{X}\mid\doop(\mathbf{T}=\mathbf{t})}(\mathbf{x})\propto\widetilde{b}_{\mathbf{t}}(\mathbf{x})e^{\widetilde{\eta}_{\mathbf{t}}(\theta)^\top\widetilde{s}_{\mathbf{t}}(\mathbf{x})}$, 
    where $\widetilde{b}_{\mathbf{t}}$, $\widetilde{\eta}_{\mathbf{t}}$,
    and $\widetilde{s}_{\mathbf{t}}$ are such that
    the parametrization is minimal, full, 
    and $\widetilde{\eta}_{\mathbf{t}}(\Theta)$ is open.
    We now adapt the argument used in the proof of Theorem~8 
    of \citet{Boeken2026} and proceed by showing the following steps. 
    \begin{enumerate}[label=(\alph*)]
        \item\label{itm:step1} For all $V\in\mathbf{V}$, 
        $v\in\mathcal{X}_V$, and 
        $\mathbf v_{\operatorname{pa}(V)\setminus\mathbf T}\in\mathcal{X}_{\operatorname{pa}(V)\setminus\mathbf T}$, 
        the real-valued map $\theta_V\mapsto p_{\theta_V}(v\mid \mathbf v_{\operatorname{pa}(V)})\big|_{\,\mathbf T=\mathbf t}$ 
        is analytic on $\Theta_V$.
        \item\label{itm:step2} For all $\mathbf y\in\mathcal{X}_{\mathbf Y}$
        and $\mathbf{w}\in\mathcal{X}_{\mathbf{W}}$, 
        the real-valued map $\theta\mapsto p_{\theta,(\mathbf{Y},\mathbf{W})\mid\doop(\mathbf{T}=\mathbf{t})}(\mathbf y,\mathbf{w})$ 
        is analytic on $\Theta$.
        \item\label{itm:step3} The $\mathbb{R}^{k_{\mathbf{t}}}$-valued natural-parameter map 
        $\theta\mapsto \widetilde{\eta}_{\mathbf{t}}(\theta)$ associated with the regular exponential family in \zcref[S]{ass:regular} is analytic on $\Theta$. 
        \item\label{itm:step4} For all $\mathbf y\in\mathcal{X}_{\mathbf Y}$, 
        the real-valued map $\theta\mapsto p_{\theta,\mathbf{Y}\mid\doop(\mathbf{T}=\mathbf{t})}(\mathbf y)$ is analytic on $\Theta$.
    \end{enumerate}
    After showing the above steps, since $\mathbf{t}\in\mathcal{X}_{\mathbf{T}}$ was arbitrary, the result follows.
    
    \paragraph{Step~\ref{itm:step1}.}
    Fix $V\in\mathbf{V}$, $v\in\mathcal{X}_V$, and 
    $\mathbf v_{\operatorname{pa}(V)\setminus{\mathbf T}}\in\mathcal{X}_{\operatorname{pa}(V)\setminus\mathbf T}$. 
    Let $\lambda_V\coloneqq b_V^*(\cdot,\mathbf v_{\operatorname{pa}(V)\setminus\mathbf T},\mathbf t)\mu_V$
    be a weighted measure on $\mathcal{X}_V$.
    Define the measure $\nu_V$ on $\mathbb{R}^{k_V}$ as the 
    pushforward of $\lambda_V$ under the map
    $s_V^*(\cdot,\mathbf v_{\operatorname{pa}(V)\setminus\mathbf T},\mathbf t):\mathcal{X}_V\to\mathbb{R}^{k_V}$, that is,
    $\nu_V\coloneqq (s_V^*(\cdot,\mathbf v_{\operatorname{pa}(V)\setminus\mathbf T},\mathbf t))_\sharp \lambda_V$.
    Then, for all $\theta_V\in\Theta_V$,
    \begin{align*}
        e^{A_V^*(\eta_V(\theta_V),\mathbf v_{\operatorname{pa}(V)\setminus\mathbf T},\mathbf t)} &= 
        \int_{\mathcal{X}_V} b_V^*(v,\mathbf v_{\operatorname{pa}(V)\setminus\mathbf T},\mathbf t)e^{\eta_V(\theta_V)^\top s_V^*(v,\mathbf v_{\operatorname{pa}(V)\setminus\mathbf T},\mathbf t)}\mathrm{d}\mu_V(v)\\
        &=\int_{\mathbb{R}^{k_V}}e^{\eta_V(\theta_V)^\top \mathbf{s}}\mathrm{d}\nu_V(\mathbf{s}).
    \end{align*}
    
    Let $\mathcal{N}_V\coloneqq\{\gamma_V\in\mathbb{R}^{k_V}:\int_{\mathbb{R}^{k_V}} e^{\gamma_V^\top\mathbf{s}}\mathrm{d}\nu_V(\mathbf{s})<\infty\}$
    and define $G_V:\mathcal{N}_V\to(0,\infty)$
    by $\gamma_V\mapsto\int_{\mathbb{R}^{k_V}}e^{\gamma_V^\top \mathbf{s}}\mathrm{d}\nu_V(\mathbf{s})$.
    In particular, for all $\theta_V\in\Theta_V$, 
    $G_V(\eta_V(\theta_V))=e^{A_V^*(\eta_V(\theta_V),\mathbf v_{\operatorname{pa}(V)\setminus\mathbf T},\mathbf t)}$.
    Now define $\mathcal{T}_V\coloneqq\{\gamma_V+i\zeta_V:\gamma_V\in\mathcal{N}_V,\zeta_V\in\mathbb{R}^{k_V}\}$ and define the complex extension 
    $\widetilde{G}_V:\mathcal{T}_V\to\mathbb{C}$
    by $z_V\mapsto\int_{\mathbb{R}^{k_V}}e^{z_V^\top \mathbf{s}}\mathrm{d}\nu_V(\mathbf{s})$.
    Since, 
    for all $\gamma_V\in\mathcal{N}_V$, $\zeta_V\in\mathbb{R}^{k_V}$, $z_V=\gamma_V+i\zeta_V$, 
    and $\mathbf{s}\in\mathbb{R}^{k_V}$,
    one has 
    $|e^{z_V^\top \mathbf{s}}|=e^{\gamma_V^\top \mathbf{s}}$, 
    $\widetilde{G}_V$ is the Fourier-Laplace transform of the pushforward 
    measure $\nu_V$. By Theorem~7.2 of \citet{BarndorffNielsen2014},
    the complex-valued map $z_V\mapsto\widetilde{G}_V(z_V)$ is analytic on the interior
    $\mathrm{int}(\mathcal{T}_V)$. Since 
    $\eta_V(\Theta_V)$ is open by 
    \ifsupplementary
        Assumption~2,
    \else
        \zcref[S]{ass:eta-analytic}, 
    \fi
    and $\eta_V(\Theta_V)\subseteq\mathcal N_V$, we have
    $\eta_V(\Theta_V)\subseteq\mathrm{int}(\mathcal N_V)$.
    Hence, for all $\theta_V\in\Theta_V$, 
    $\eta_V(\theta_V)+i0\in\mathrm{int}(\mathcal T_V)$.
    Moreover, the map $\theta_V\mapsto\eta_V(\theta_V)$ is 
    analytic on $\Theta_V$ by \zcref[S]{ass:eta-analytic},
    so the composition
    \begin{equation}\label{eq:analytic-second-term}
        \theta_V\mapsto\widetilde{G}_V(\eta_V(\theta_V))=e^{A_V^*(\eta_V(\theta_V),\mathbf v_{\operatorname{pa}(V)\setminus\mathbf T},\mathbf t)}
    \end{equation}
    is analytic on $\Theta_V$. 
    Since, for all $\theta_V\in\Theta_V$, 
    $\widetilde{G}_V(\eta_V(\theta_V)) > 0$,
    its reciprocal is analytic as well.

    Moreover, since $s_V^*(v,\mathbf v_{\operatorname{pa}(V)\setminus\mathbf T},\mathbf t)$ 
    does not depend on $\theta_V$ and the exponential of an 
    analytic function is analytic, it follows that the real-valued map 
    \begin{equation}\label{eq:analytic-first-term}
        \theta_V\mapsto e^{\eta_V(\theta_V)^\top s_V^*(v,\mathbf v_{\operatorname{pa}(V)\setminus\mathbf T},\mathbf t)}
    \end{equation}
    is analytic on $\Theta_V$.
    
    Therefore, the map 
    $\theta_V\mapsto p_{\theta_V}(v\mid \mathbf v_{\operatorname{pa}(V)})\big|_{\,\mathbf T=\mathbf t}$
    is analytic on $\Theta_V$ since it is the product of the 
    $\theta_V$-independent factor 
    $b_V^*(v,\mathbf v_{\operatorname{pa}(V)\setminus\mathbf T},\mathbf t)$, 
    the analytic function in 
    \zcref[noname]{eq:analytic-first-term}, and the reciprocal of the analytic 
    function in \zcref[noname]{eq:analytic-second-term}. 
    Since $V\in\mathbf{V}$, $v\in\mathcal{X}_V$, 
    and $\mathbf v_{\operatorname{pa}(V)\setminus\mathbf T}\in\mathcal{X}_{\operatorname{pa}(V)\setminus\mathbf T}$ 
    were arbitrary, this holds for all such values. 
    This proves~\ref{itm:step1}.

    \paragraph{Step~\ref{itm:step2}.}
    By~\zcref[noname]{eq:truncated-factorization}, 
    for all $\mathbf y \in\mathcal{X}_{\mathbf Y}$
    and $\mathbf{w}\in\mathcal{X}_{\mathbf{W}}$, 
    the real-valued map $\theta\mapsto p_{\theta,(\mathbf{Y},\mathbf{W})\mid\doop(\mathbf{T}=\mathbf{t})}(\mathbf y,\mathbf{w})$ 
    is a finite product of analytic functions, hence analytic on $\Theta$. 
    Therefore, \ref{itm:step2} holds.

    \paragraph{Step~\ref{itm:step3}.}
    By 
    \ifsupplementary
        Assumption~3,
    \else
        \zcref[S]{ass:regular}, 
    \fi
    the post-intervention joint distribution lies in a regular exponential family. 
    Then, for all $\theta\in\Theta$ and $\mathbf{x}\in\mathcal{X}_{\mathbf{X}}$, 
    $$p_{\theta,\mathbf{X}\mid\doop(\mathbf{T}=\mathbf{t})}(\mathbf{x})=\widetilde{b}_{\mathbf{t}}(\mathbf{x})e^{\widetilde{\eta}_{\mathbf{t}}(\theta)^\top\widetilde{s}_{\mathbf{t}}(\mathbf{x})-\widetilde{A}_{\mathbf{t}}(\widetilde{\eta}_{\mathbf{t}}(\theta))}.$$ 
    Minimality guarantees the existence of $k+1$ points 
    $\mathbf{x}_0,\dots,\mathbf{x}_k\in\mathcal{X}_{\mathbf{X}}$
    such that the $k$ difference vectors
    $\mathbf{u}_1^{\mathbf{t}}:=\widetilde{s}_{\mathbf{t}}(\mathbf{x}_1)-\widetilde{s}_{\mathbf{t}}(\mathbf{x}_0),\ldots,\mathbf{u}_k^{\mathbf{t}}:=\widetilde{s}_{\mathbf{t}}(\mathbf{x}_k)-\widetilde{s}_{\mathbf{t}}(\mathbf{x}_0)$ are linearly independent.
    For all $\theta\in\Theta$ and $\mathbf{x}\in\mathcal{X}_{\mathbf{X}}$,
    taking the logarithm of the density yields 
    $\log{p_\theta(\mathbf{x}\mid\operatorname{do}(\mathbf t))}=\log{\widetilde{b}_{\mathbf{t}}(\mathbf{x})}+\widetilde{\eta}_{\mathbf{t}}(\theta)^\top\widetilde{s}_{\mathbf{t}}(\mathbf{x})-\widetilde{A}_{\mathbf{t}}(\widetilde{\eta}_{\mathbf{t}}(\theta))$. 
    Thus, for all $i\in\{1,\dots,k\}$,
    \begin{align*}
        \log{p_{\theta,\mathbf{X}\mid\doop(\mathbf{T}=\mathbf{t})}(\mathbf{x}_i)}-\log{p_{\theta,\mathbf{X}\mid\doop(\mathbf{T}=\mathbf{t})}(\mathbf{x}_0)}&=
        (\log{\widetilde{b}_{\mathbf{t}}(\mathbf{x}_i)}-\log{\widetilde{b}_{\mathbf{t}}(\mathbf{x}_0)})+\widetilde{\eta}_{\mathbf{t}}(\theta)^\top(\widetilde{s}_{\mathbf{t}}(\mathbf{x}_i)-\widetilde{s}_{\mathbf{t}}(\mathbf{x}_0))\\
        &=-C_i^{\mathbf{t}}+(\mathbf{u}_i^{\mathbf{t}})^\top\widetilde{\eta}_{\mathbf{t}}(\theta),
    \end{align*}
    where $C_i^{\mathbf{t}}\coloneqq-\log{\widetilde{b}_{\mathbf{t}}(\mathbf{x}_i)}+\log{\widetilde{b}_{\mathbf{t}}(\mathbf{x}_0)}$. 
    We rearrange this to define, for $i\in\{1,\dots,k\}$,
    \begin{equation*}
        h_i^{\mathbf{t}}(\theta)\coloneqq\log{p_{\theta,\mathbf{X}\mid\doop(\mathbf{T}=\mathbf{t})}(\mathbf{x}_i)}-\log{p_{\theta,\mathbf{X}\mid\doop(\mathbf{T}=\mathbf{t})}(\mathbf{x}_0)}+C_i^{\mathbf{t}}=(\mathbf{u}_i^{\mathbf{t}})^\top\widetilde{\eta}_{\mathbf{t}}(\theta).
    \end{equation*}
    Since in the previous steps we established that the joint density
    is analytic in $\theta$, for all $i\in\{1,\dots,k\}$, $h_i^{\mathbf{t}}$ is
    analytic as well as it is a difference of analytic functions
    plus a constant.

    Let $U_{\mathbf{t}} \coloneqq [\mathbf{u}_1^{\mathbf{t}}, \dots, \mathbf{u}_k^{\mathbf{t}}]$ be the 
    $k \times k$ matrix whose $i$-th column is $\mathbf{u}_i^{\mathbf{t}}$, 
    for all $i\in\{1,\dots,k\}$, and let 
    $h_{\mathbf{t}}(\theta) \coloneqq (h_1^{\mathbf{t}}(\theta), \dots, h_k^{\mathbf{t}}(\theta))^\top$
    be the $k$-dimensional vector whose $i$-th entry is $h_i^{\mathbf{t}}(\theta)$, for all $i\in\{1,\dots,k\}$. 
    Then, we can write $U_{\mathbf{t}}^\top\widetilde{\eta}_{\mathbf{t}}(\theta) = h_{\mathbf{t}}(\theta)$.
    Since the vectors $\mathbf{u}_1^{\mathbf{t}}, \dots, \mathbf{u}_k^{\mathbf{t}}$ are linearly
    independent, the matrix $U_{\mathbf{t}}$ (and therefore $U_{\mathbf{t}}^\top$) has full rank
    $k_{\mathbf{t}}$ and is invertible. Multiplying both sides by the inverse 
    yields $\widetilde{\eta}_{\mathbf{t}}(\theta) = (U_{\mathbf{t}}^\top)^{-1}h_{\mathbf{t}}(\theta)$. 
    As a linear combination of the analytic components of $h_{\mathbf{t}}(\theta)$,
    the map $\theta \mapsto \widetilde{\eta}_{\mathbf{t}}(\theta)$ is analytic on $\Theta$. 
    The claim in~\ref{itm:step3} then follows.

    \paragraph{Step~\ref{itm:step4}.}
    Let $\mathbf A,\mathbf B$ be such that $\mathbf A\sqcup \mathbf B=\mathbf V\setminus\mathbf T$.
    For all $\mathbf{b}\in\mathcal{X}_{\mathbf{B}}$, the integral
    \begin{equation*}
        \widetilde\eta_{\mathbf{t}} \mapsto \int_{\mathcal{X}_{\mathbf{A}}}\widetilde{b}_{\mathbf{t}}(\mathbf{a},\mathbf{b})e^{\widetilde\eta_{\mathbf{t}}^\top\widetilde{s}_{\mathbf{t}}(\mathbf{a},\mathbf{b})}\mathrm{d}\mu_{\mathbf{A}}(\mathbf{a})
    \end{equation*}
    can be viewed as the Fourier-Laplace transform of a 
    pushforward measure. By Theorem~7.2 of 
    \citet{BarndorffNielsen2014}, this map is analytic on a 
    complex domain which, 
    due to the openness of $\widetilde{\eta}_{\mathbf{t}}(\Theta)$ given by
    \ifsupplementary
        Assumption~3,
    \else
        \zcref[S]{ass:regular}, 
    \fi
    is open and contains $\widetilde{\eta}_{\mathbf{t}}(\Theta)$. 
    Consequently, its restriction to the real parameter space
    is analytic on $\widetilde{\eta}_{\mathbf{t}}(\Theta)$. 
    Because the map $\theta\mapsto\widetilde{\eta}_{\mathbf{t}}(\theta)$ is analytic
    on $\Theta$ by~\ref{itm:step3}, 
    for all $\mathbf{b}\in\mathcal{X}_{\mathbf{B}}$, 
    the composition
    \begin{equation*}
        \theta\mapsto\int_{\mathcal{X}_{\mathbf{A}}}\widetilde{b}_{\mathbf{t}}(\mathbf{a},\mathbf{b})e^{\widetilde\eta_{\mathbf{t}}(\theta)^\top\widetilde{s}_{\mathbf{t}}(\mathbf{a},\mathbf{b})}\mathrm{d}\mu_{\mathbf{A}}(\mathbf{a})
    \end{equation*}
    is analytic on $\Theta$.
    
    Taking $\mathbf{A}=\mathbf{W}=\mathbf{V}\setminus(\mathbf T\sqcup\mathbf Y)$,
    we obtain, for all $\theta\in\Theta$, $\mathbf{t}\in\mathcal{X}_{\mathbf{T}}$, and $\mathbf y\in\mathcal{X}_{\mathbf Y}$, 
    the unnormalized marginal
    \begin{equation*}
        p_{\theta,\mathbf{Y}\mid\doop(\mathbf{T}=\mathbf{t})}(\mathbf y)
        \propto
        \int_{\mathcal X_{\mathbf W}}
        \widetilde b_{\mathbf{t}}(\mathbf y,\mathbf w)
        e^{\widetilde\eta_{\mathbf{t}}(\theta)^\top\widetilde s_{\mathbf{t}}(\mathbf y,\mathbf w)}
        \mathrm d\mu_{\mathbf W}(\mathbf w).
    \end{equation*}
    Taking instead $\mathbf{A}=\mathbf{X}=\mathbf{V}\setminus\mathbf{T}$, we obtain, 
    for all $\theta\in\Theta$, the normalizing factor
    \begin{equation*}
        e^{\widetilde{A}_{\mathbf{t}}(\widetilde{\eta}_{\mathbf{t}}(\theta))}=\int_{\mathcal{X}_{\mathbf{X}}}\widetilde{b}_{\mathbf{t}}(\mathbf{x})e^{\widetilde\eta_{\mathbf{t}}(\theta)^\top\widetilde{s}_{\mathbf{t}}(\mathbf{x})}\mathrm{d}\mu_{\mathbf{X}}(\mathbf{x}).
    \end{equation*}
    Both the unnormalized marginal and the normalizing factor are therefore analytic in $\theta$. 
    Since, for all $\theta \in \Theta$, 
    $e^{\widetilde{A}_{\mathbf{t}}(\widetilde{\eta}_{\mathbf{t}}(\theta))} > 0$, 
    its reciprocal $e^{-\widetilde{A}_{\mathbf{t}}(\widetilde{\eta}_{\mathbf{t}}(\theta))}$
    is also analytic in $\theta$. 
    Therefore, for all $\mathbf y\in\mathcal{X}_{\mathbf Y}$ and all $\mathbf{t}\in\mathcal{X}_{\mathbf{T}}$, 
    the real-valued map $\theta\mapsto p_{\theta,\mathbf{Y}\mid\doop(\mathbf{T}=\mathbf{t})}(\mathbf y)$
    is analytic on $\Theta$ as a product of analytic functions. 
    This concludes the proof of 
    \ifsupplementary
        Proposition~1.
    \else
        \zcref[S]{prop:measure-zero}.
    \fi
\end{proof}

\subsection{Proof of 
\ifsupplementary
    Theorem~2
\else
    \zcref[S]{thm:as-correct-verifier}
\fi
}\label{app:proof-as-correct-verifier}

\ascorrectverifier*

\begin{proof}
If the input to the falsifier is $\phi=\texttt{none}$, 
the falsifier checks non-identifiability using the ID algorithm,
returning $\texttt{false}$ if the ID algorithm returns an identifying formula and $\texttt{true}$ else. 
Since the ID algorithm is sound and complete for identification, the falsifier's output for the input $\texttt{none}$ is correct.

If the input to the falsifier is a candidate observational formula
$\phi\in\Phi_{\mathbf Y,\mathbf{T}}$
there are two cases.
First, if $\phi$ is identifying relative to 
the chosen parametric family, then
\ifsupplementary
    Equation~(1)
\else
    \zcref[S]{eq:identifying-formula} 
\fi
holds for all $\mathbf{t}\in\mathcal{X}_{\mathbf{T}}$ and for all densities factorizing according to the graph whenever 
$\phi$ is well-defined. Hence, the falsifier never finds a 
counterexample and outputs $\texttt{true}$; in other words,
the falsifier never incorrectly rejects an observational formula that is identifying as $\texttt{false}$. 
Second, if
$\phi$ is not identifying
relative to the chosen 
parametric family, then, by 
\ifsupplementary
    Proposition~1,
\else
    \zcref[S]{prop:measure-zero},
\fi
the set $S$ of parameters where the candidate formula equals
the target interventional density has Lebesgue measure zero. 
By absolute continuity of $\pi$ with respect to the Lebesgue measure, this implies that $\pi(S)=0$. Then,
$$\mathbb{P}\left(\bigcap_{i=1}^K\{\theta_i\in S\}\right)=\prod_{i=1}^K\mathbb{P}(\theta_i\in S)=\left(\pi(S)\right)^K=0,$$
which implies that, if $\phi$ is not identifying relative to the chosen
parametric family, the falsifier incorrectly accepts it as $\texttt{true}$ with probability zero.

This completes the proof of \zcref[S]{thm:as-correct-verifier}, 
which establishes a one-sided guarantee: if the falsifier
rejects a candidate formula, it has found an actual counterexample
and the rejection is correct; if the falsifier accepts a 
candidate formula, its output is correct almost surely,
relative to the chosen parametric family.
\end{proof}

\subsection{Proof of \zcref[S]{thm:bound}}\label{app:proof-bound}

\begin{restatable*}
[Bound on false acceptance of non-identifying formulas]
{theorem}{falseacceptancebound}
\label{thm:bound}
    Suppose that the evaluation length of $\phi$ is at most $c$,
    and that $\phi$ is non-identifying relative to 
    $\mathcal{P}_{\Theta}(\mathcal{X}_{\mathbf{V}})$, that is,
    there exist $\theta^*\in\Theta$ and $\mathbf{t}^*\in\mathcal{X}_{\mathbf{T}}$ such that
    the equality
    $\llbracket \phi\rrbracket(p_{\theta^*,\mathbf O},\mathbf t^*)= p_{\theta^*,\mathbf Y\mid\doop(\mathbf T=\mathbf t^*)}$ 
    $\mu_{\mathbf{Y}}$-almost everywhere does not hold.
    Let $\theta^\prime$ be the vector that collects all sampled
    parameters.
    Then, the probability that the falsifier accepts $\phi$
    at the sampled parameter $\theta^\prime$ satisfies
    \begin{equation}\label{eq:sharper-bound}
        \mathbb P\left(
        \forall\mathbf t\in\mathcal X_{\mathbf T},
        \llbracket \phi\rrbracket(p_{\theta^\prime,\mathbf O},\mathbf t)
        =
        p_{\theta^\prime,\mathbf Y\mid \doop(\mathbf T=\mathbf t)}
        \ \ \mu_{\mathbf Y}\text{-a.e.}
        \right)
        \le
        \frac{
        (2n-1)\left(1+\prod_{l=1}^c\gamma_l\right)
        }{
        \min_{1\le i\le d}|A_i|
        },
    \end{equation}
    where, for all $l\in\{1,\dots,c\}$,
    \begin{equation}\label{eq:gamma}
        \gamma_l
        \coloneqq
        \begin{cases}
            1, &\text{if operation $l$ is Gaussian marginalization},\\
            2, &\text{if operation $l$ is scalar addition, subtraction, multiplication or division},\\
            n+1, &\text{if operation $l$ is Gaussian conditioning.}
            \end{cases}
    \end{equation}
\end{restatable*}

\begin{proof}
Since $\phi$ is non-identifying and since both the distribution
returned by $\phi$ and the target interventional distribution
are Gaussian, there exist 
$\theta^*\in\Theta$ and $\mathbf t^*\in\mathcal X_{\mathbf T}$
such that at least one scalar entry of either 
$\mu^\phi(\theta^*,\mathbf{t}^*)-\mu^{\doop}(\theta^*,\mathbf{t}^*)$
or $\Sigma^\phi(\theta^*)-\Sigma^{\doop}(\theta^*)$ 
is non-zero.
Let $\Delta(\theta,\mathbf{t}^*)$ denote such an entry,
viewed as a rational function of $\theta$ after fixing 
$\mathbf{t}^*$, and write 
$\Delta(\theta,\mathbf{t}^*)=p_{\Delta}(\theta)/q_{\Delta}(\theta)$.
For all $\theta\in\Theta$, the corresponding Gaussian model
is non-degenerate, and the admissible operations in $\phi$
preserve non-degeneracy (see \zcref[S]{app:gaussian-closure});
hence, $q_\Delta(\theta)\neq0$.
Since $\Delta(\theta^*,\mathbf t^*)\neq0$, we have
$p_{\Delta}(\theta^*)\neq0$, and hence $p_{\Delta}\not\equiv0$.
Let
\[
E\coloneqq
\left\{
\forall\mathbf t\in\mathcal X_{\mathbf T},\ 
\llbracket \phi\rrbracket(p_{\theta^\prime,\mathbf O},\mathbf t)
=
p_{\theta^\prime,\mathbf Y\mid \doop(\mathbf T=\mathbf t)}
\ \ \mu_{\mathbf Y}\text{-a.e.}
\right\}
\]
be the event that the falsifier accepts $\phi$ at the
sampled parameter $\theta^\prime$.
If $E$ occurs, then the falsifier accepts, meaning that
all scalar mean and covariance 
discrepancies vanish for all intervention values.
In particular, the specific witness discrepancy
$\Delta(\theta^\prime,\mathbf t^*)$ evaluated at the 
sampled parameter must vanish.
Therefore,
\begin{equation*}
    E\subseteq\{p_{\Delta}(\theta^\prime)=0\}
    \quad\text{ and }\quad
    \mathbb P(E)\le \mathbb P(p_{\Delta}(\theta^\prime)=0).
\end{equation*}
By \zcref[S]{thm:schwartz-zippel}, it remains to bound
the degree of the non-zero polynomial $p_{\Delta}$.

We first bound the degrees of the Gaussian quantities from 
which the evaluation of $\phi$ starts. 
Recall that $B$ is the $n\times n$ matrix collecting all edge
coefficients.
Since $\mathcal{G}$ is
acyclic, for all integer $m\ge n$, $(B^\top)^m=0$, and therefore
\begin{equation*}
    (I-B^\top)^{-1}=I+\sum_{m=1}^{n-1}(B^\top)^m.
\end{equation*}
Hence each entry of $(I-B^\top)^{-1}$ is a polynomial in the edge
coefficients of degree at most $n-1$. 
For all $i,j\in\{1,\dots,n\}$, the $(i,j)$-th entry of the covariance matrix in \zcref[S]{eq:mean-covariance} is
\begin{equation*}
    \Sigma^{\mathrm{full}}_{ij}(\theta)=
    \sum_{k=1}^n[(I-B^\top)^{-1}]_{ik}\sigma_k^2[(I-B)^{-1}]_{kj},
\end{equation*}
where, for all $k\in\{1,\dots,n\}$, $\sigma_k^2$ is the 
variance of the $k$-th variable conditional on its parents.
Therefore, the corresponding degree is 
such that
\begin{equation*}
    \deg(\Sigma^{\mathrm{full}}_{ij}(\theta))\le(n-1)+1+(n-1)=2n-1.
\end{equation*}
Similarly, for all $i\in\{1,\dots,n\}$,
the $i$-th entry of the mean is
$\mu^{\mathrm{full}}_i(\theta)=
\sum_{k=1}^n[(I-B^\top)^{-1}]_{ik}\alpha_k$,
where, for all $k\in\{1,\dots,n\}$, $\alpha_k$ is the mean of the
$k$-th noise term. The corresponding degree is such that 
\begin{equation*}
    \deg(\mu^{\mathrm{full}}_i(\theta))\le (n-1)+1=n\le 2n-1.
\end{equation*}
Set $h_0\coloneqq2n-1$.
Every scalar entry of the mean and covariance
has therefore degree at most $h_0$ as polynomial in the parameters.
Equivalently, the observational Gaussian quantities can be
represented in shared-denominator form as
$\mu^{\mathrm{full}}(\theta)=m^{\mathrm{full}}(\theta)/1$ and
$\Sigma^{\mathrm{full}}(\theta)=A^{\mathrm{full}}(\theta)/1$,
where all numerator entries have degree at most $h_0$.
The same bound applies to the entries of
the target interventional quantities 
$\mu^{\doop}(\theta,\mathbf t^*)$ and $\Sigma^{\doop}(\theta)$,
because after intervention these entries are again 
obtained from a linear Gaussian submodel on at most $n$ nodes.

We now bound how degrees can grow during the evaluation of the
candidate formula $\phi$. 
Suppose that, at some intermediate stage, scalar rational
expressions have numerator and denominator degree at most $h$,
and every intermediate Gaussian mean and covariance is represented
in shared-denominator form as
\[
    \mu(\theta)=\frac{m(\theta)}{e(\theta)},
    \qquad
    \Sigma(\theta)=\frac{A(\theta)}{d(\theta)},
\]
where $m(\theta)$ is a vector of polynomial numerators,
$A(\theta)$ is a matrix of polynomial numerators, and
$e(\theta)$ and $d(\theta)$ are scalar polynomial denominators.
Assume that every entry of $m(\theta)$ and $A(\theta)$, and the
denominators $e(\theta)$ and $d(\theta)$, have degree at most $h$.
We claim that one primitive operation increases this degree
bound by at most the factor $n+1$.

Marginalization only selects subvectors and submatrices, and 
therefore does not increase degrees. 
Scalar addition, subtraction,
multiplication, and division increase $h$ by at most a factor of $2$. 
For example,
$p_1(\theta)/q_1(\theta)+p_2(\theta)/q_2(\theta)
=(p_1(\theta)q_2(\theta)+p_2(\theta)q_1(\theta))/
(q_1(\theta)q_2(\theta))$, 
and both numerator and denominator degrees are at most $2h$.
The same holds for multiplication and division.

Consider an intermediate Gaussian distribution on disjoint
variable sets $(\mathbf A,\mathbf B)$, 
with $|\mathbf{A}|=q\le n$ and $|\mathbf{B}|=p\le n$.
Write its mean vector and covariance matrix in shared-denominator
form as
\[
\mu(\theta)
=
\frac{1}{e(\theta)}
\begin{pmatrix}
m_{\mathbf A}(\theta)\\
m_{\mathbf B}(\theta)
\end{pmatrix},
\qquad
\Sigma(\theta)
=
\frac{1}{d(\theta)}
\begin{pmatrix}
A_{\mathbf A\mathbf A}(\theta) &
A_{\mathbf A\mathbf B}(\theta)\\
A_{\mathbf B\mathbf A}(\theta) &
A_{\mathbf B\mathbf B}(\theta)
\end{pmatrix}.
\]
where $\Sigma(\theta),A_{\mathbf{AA}}(\theta),A_{\mathbf{BB}}(\theta)$ are symmetric
and
$A_{\mathbf B\mathbf A}(\theta)=A_{\mathbf A\mathbf B}(\theta)^\top$.
For all $\mathbf{b}\in\mathcal{X}_{\mathbf{B}}$,
conditioning gives
\begin{align*}
    \mu_{\mathbf A\mid \mathbf B=\mathbf b}(\theta)
    &=
    \mu_{\mathbf A}(\theta)
    +
    \Sigma_{\mathbf A\mathbf B}(\theta)
    \Sigma_{\mathbf B\mathbf B}(\theta)^{-1}
    (\mathbf b-\mu_{\mathbf B}(\theta)),\\
    \Sigma_{\mathbf A\mid \mathbf B}(\theta)
    &=
    \Sigma_{\mathbf A\mathbf A}(\theta)
    -
    \Sigma_{\mathbf A\mathbf B}(\theta)
    \Sigma_{\mathbf B\mathbf B}(\theta)^{-1}
    \Sigma_{\mathbf B\mathbf A}(\theta).
\end{align*}

First, we treat the covariance update. Since
$\Sigma_{\mathbf B\mathbf B}(\theta)=
A_{\mathbf B\mathbf B}(\theta)/d(\theta)$, we have
\[
    \Sigma_{\mathbf B\mathbf B}(\theta)^{-1}
    =
    d(\theta)
    \frac{\operatorname{adj}(A_{\mathbf B\mathbf B}(\theta))}
    {\det(A_{\mathbf B\mathbf B}(\theta))}.
\]
We bound the degree of the determinant and the adjugate
entries. By the Leibniz formula,
\begin{equation*}
    \det(A_{\mathbf B\mathbf B}(\theta))
    =
    \sum_{\lambda\in \Lambda_p}
    \operatorname{sign}(\lambda)
    \prod_{\ell=1}^p
    (A_{\mathbf B\mathbf B}(\theta))_{\ell,\lambda(\ell)},
\end{equation*}
where $\Lambda_p$ is the set of all permutations of $\{1,\dots,p\}$.
For all $\lambda\in \Lambda_p$, the product contains $p$ factors, each
of degree at most $h$. Hence, each product has degree at most
$ph$. Taking a sum cannot increase the degree beyond the maximum
degree of the summands, so
\[
    \deg(\det(A_{\mathbf B\mathbf B}(\theta)))\le ph.
\]
For all $i,j\in\{1,\dots,p\}$, the $(i,j)$-th entry of the
adjugate matrix is a cofactor:
\begin{equation*}
    \operatorname{adj}(A_{\mathbf B\mathbf B}(\theta))_{ij}
    =
    (-1)^{i+j}
    \det\!\left((A_{\mathbf B\mathbf B}(\theta))^{(j,i)}\right),
\end{equation*}
where $(A_{\mathbf B\mathbf B}(\theta))^{(j,i)}$ is obtained by
deleting row $j$ and column $i$. This minor has size 
$(p-1)\times(p-1)$. Applying the same determinant argument to this
minor, each determinant term is a product of $p-1$ entries, each
of degree at most $h$. Hence each such product has degree at most
$(p-1)h$, and taking the sum over permutations cannot increase
the degree. Therefore,
\[
    \deg(\operatorname{adj}(A_{\mathbf B\mathbf B}(\theta))_{ij})
    \le
    (p-1)h.
\]

Substituting these expressions into the conditional covariance
gives
\[
\Sigma_{\mathbf A\mid \mathbf B}(\theta)
=
\frac{
A_{\mathbf A\mathbf A}(\theta)\det(A_{\mathbf B\mathbf B}(\theta))
-
A_{\mathbf A\mathbf B}(\theta)
\operatorname{adj}(A_{\mathbf B\mathbf B}(\theta))
A_{\mathbf B\mathbf A}(\theta)
}{
d(\theta)\det(A_{\mathbf B\mathbf B}(\theta))
}.
\]
Each entry of the first numerator term has degree at most
$h+ph=(p+1)h$. 
For the second numerator term, for all $i,j\in\{1,\dots,q\}$,
the $(i,j)$-th entry is a sum of products of the form
\[
(A_{\mathbf A\mathbf B}(\theta))_{iu}
\operatorname{adj}(A_{\mathbf B\mathbf B}(\theta))_{uv}
(A_{\mathbf B\mathbf A}(\theta))_{vj},
\]
with $u,v\in\{1,\dots,p\}$. Each such product has degree at most
$h+(p-1)h+h=(p+1)h$.
Since taking sums cannot increase the degree beyond the maximum
degree of the summands, every numerator entry has degree at most
$(p+1)h$. The denominator
$d(\theta)\det(A_{\mathbf B\mathbf B}(\theta))$ also has degree
at most $h+ph=(p+1)h$. Since $p\le n$, every entry of
$\Sigma_{\mathbf A\mid \mathbf B}(\theta)$ can be represented with
numerator and denominator degree at most $(n+1)h$.

The conditional mean is controlled by the same inverse block.
Since $\mathbf b$ is fixed, each component of
$\mathbf b-\mu_{\mathbf B}(\theta)$ can be written as
$(\mathbf b e(\theta)-m_{\mathbf B}(\theta))/e(\theta)$,
with numerator and denominator degree at most $h$.
Using the expression for
$\Sigma_{\mathbf B\mathbf B}(\theta)^{-1}$ above, we obtain
\[
\mu_{\mathbf A\mid \mathbf B=\mathbf b}(\theta)
=
\frac{
m_{\mathbf A}(\theta)\det(A_{\mathbf B\mathbf B}(\theta))
+
A_{\mathbf A\mathbf B}(\theta)
\operatorname{adj}(A_{\mathbf B\mathbf B}(\theta))
(\mathbf b e(\theta)-m_{\mathbf B}(\theta))
}{
e(\theta)\det(A_{\mathbf B\mathbf B}(\theta))
}.
\]
The first numerator term has degree at most
$h+ph=(p+1)h$. Each entry of the second numerator term is a sum
of products with degree at most
$h+(p-1)h+h=(p+1)h$. The denominator has degree at most
$h+ph=(p+1)h$. Thus, every entry of the conditional mean also
has numerator and denominator degree at most $(n+1)h$.

We showed that every primitive operation increases the current
degree bound by at most the factor $n+1$. Since the evaluation
length of $\phi$ is at most $c$, if, for all
$l\in\{0,\dots,c-1\}$, $h_l$ denotes the degree bound after $l$
primitive operations, then $h_{l+1}\le \gamma_{l+1}h_l$,
with $\gamma_{l+1}$ defined in \zcref[S]{eq:gamma}. Starting from
$h_0=2n-1$, we obtain $h_c\le(2n-1)\prod_{l=1}^c\gamma_l.$
Hence, every scalar entry of 
$\mu^\phi(\theta,\mathbf t^*)$ and $\Sigma^\phi(\theta)$ 
can be written as a rational function whose numerator and 
denominator have total degree at most
\begin{equation*}
    H\coloneqq(2n-1)\prod_{l=1}^c\gamma_l.
\end{equation*}

We now return to the polynomial $p_{\Delta}$. The discrepancy
$\Delta(\theta,\mathbf t^*)$ is the difference between 
one scalar entry produced by $\phi$ and the corresponding
scalar entry of the target interventional distribution.
Write the entry produced by $\phi$ as
$p_\phi(\theta)/q_\phi(\theta)$, where 
$\deg(p_\phi(\theta)),\deg(q_\phi(\theta))\le H$.
Write the corresponding target entry as 
$p_{\doop}(\theta)/q_{\doop}(\theta)$. By the previous
bound on the target interventional mean and covariance,
$\deg(p_{\doop}(\theta)),\deg(q_{\doop}(\theta))\le h_0=2n-1$.
We have
\[
\Delta(\theta,\mathbf t^*)
=
\frac{p_\phi(\theta)}{q_\phi(\theta)}
-
\frac{p_{\doop}(\theta)}{q_{\doop}(\theta)}
=
\frac{
p_\phi(\theta)q_{\doop}(\theta)
-
p_{\doop}(\theta)q_\phi(\theta)
}{
q_\phi(\theta)q_{\doop}(\theta)
}.
\]
Thus, we may take 
\[
p_\Delta(\theta)
=
p_\phi(\theta)q_{\doop}(\theta)
-
p_{\doop}(\theta)q_\phi(\theta).
\]
Since $\Delta(\theta^*,\mathbf t^*)\neq0$, we have
$p_\Delta(\theta^*)\neq0$, and hence $p_\Delta\not\equiv0$.
Each product has degree at most $H+h_0$, and the difference
cannot increase the degree. Hence,
$\deg(p_\Delta(\theta))\le H+h_0$.

Applying \zcref[S]{thm:schwartz-zippel} to the non-zero
polynomial $P_\Delta$ gives
\begin{equation*}
    \mathbb{P}(p_\Delta(\theta^\prime)=0)
    \le
    \frac{H+h_0}{\min_{1\le i\le d}|A_i|}
    =
    \frac{
    (2n-1)\left(1+\prod_{l=1}^c\gamma_l\right)
    }{
    \min_{1\le i\le d}|A_i|
    }.
\end{equation*}
Since $E\subseteq\{P_\Delta(\theta^\prime)=0\}$, this
concludes the proof of \zcref[S]{thm:bound}.
\end{proof}

\ifarxiv
    \appsection{Alternative implementation via the nested Markov model}\label{app:linear-sems-margs}
\else
    \section{Alternative implementation via the nested Markov model}\label{app:linear-sems-margs}
\fi

In 
\ifsupplementary
    Section~4,
\else
    \zcref[S]{sec:hiprof}, 
\fi
we instantiate the falsifier by replacing the input acyclic
directed mixed graph $\mathcal{G}^p$ with its canonical 
directed acyclic graph.
There are, however, infinitely many latent-variable 
directed acyclic graphs whose latent projection is $\mathcal{G}^p$.
An alternative is to avoid specifying the latent structure and 
to work directly with $\mathcal{G}^p$ and 
its maximal arid projection 
$\widetilde{\mathcal{G}}^p$, that is, 
a maximal arid graph on the same observed set as $\mathcal{G}^p$.
Aridity excludes certain graphical structures that prevent 
identifiability of the associated linear structural equation 
model \citep{Drton2011}. 
We refer to \citep{Shpitser2018} for a formal definition of maximal
arid projection and for an algorithm for computing it.

Let $\mathbf{X}_{\mathbf{O}}=(X_{O_1},\dots,X_{O_n})^\top$ denote 
the random vector associated with the observed variables $\mathbf{O}$.
We consider the linear Gaussian structural equation model 
associated with $\widetilde{\mathcal{G}}^p$:
\begin{equation*}
    \mathbf{X}_{\mathbf{O}}=B^\top\mathbf{X}_{\mathbf{O}}+\epsilon,
\end{equation*}
where $B=(b_{ij})\in\mathbb{R}^{n\times n}$ is such that, 
for all $i,j\in\{1,\dots,n\}$, $b_{ij}=0$ whenever $O_i\rightarrow O_j$
is not an edge of $\widetilde{\mathcal{G}}^p$, and where 
$\epsilon\sim\mathcal{N}(\boldsymbol{0},\Omega)$,
with $\Omega=(\omega_{ij})\in\mathbb{R}^{n\times n}$ positive definite 
and such that, for all $i,j\in\{1,\dots,n\}$ with $i\ne j$, 
$\omega_{ij}=0$ whenever $O_i\leftrightarrow O_j$ is not an edge of 
$\widetilde{\mathcal{G}}^p$. The resulting covariance matrix is
$\Sigma=(I_n-B)^{-\top}\Omega(I_n-B)^{-1}$. Under this 
parametrization, 
the falsifier can sample
a masked matrix $B$ and 
positive definite matrix $\Omega$, 
and then proceed as in 
\ifsupplementary
    Definition~5
\else
    \zcref[S]{def:falsifier}, 
\fi
but without choosing a particular latent structure. 
Let $\mathcal{K}\subseteq\{1,\dots,n\}$ be the set of indices 
such that $\mathbf{T}=\{O_k\mid k\in\mathcal{K}\}$. We can compute
the interventional density by setting, for all $k\in\mathcal{K}$ 
and all $i\in\{1,\dots,n\}$, $b_{ik}=0$, and setting, 
for all $k\in\mathcal{K}$ and 
$i\in\{1,\dots,n\}\setminus\{k\}$, $\omega_{ik}=\omega_{ki}=0$.

This construction is justified by the relationship between
maximal arid projections and nested Markov models.
Nested Markov models are graphical models 
for acyclic directed mixed graphs that capture not only the conditional
independences 
but also generalized equality constraints implied by
latent-variable models \citep[see][for an introduction]{Shpitser2014}.
The maximal arid projection $\widetilde{\mathcal{G}}^p$
of an acyclic directed mixed graph $\mathcal{G}^p$
defines the same nested Markov model as $\mathcal{G}^p$ 
\citep[][Theorem~30]{Shpitser2018}.
Moreover,
\citet[][Theorem~35]{Shpitser2018} show that the above linear Gaussian 
structural equation model associated with the maximal arid projection
$\widetilde{\mathcal{G}}^p$ of $\mathcal{G}^p$ coincides with the Gaussian 
nested Markov model of $\mathcal{G}^p$. 
Thus, the maximal arid projection gives a linear Gaussian
parametrization of exactly the Gaussian nested Markov model
associated with $\mathcal{G}^p$, rather than a potentially
smaller model associated with the linear Gaussian structural
equation model on $\mathcal{G}^p$ itself \citep[][Theorem~34]{Shpitser2018}.
Both implementations
(via the canonical directed acyclic graph or the maximal arid projection)
should be contrasted to specifying a particular latent-variable directed acyclic graph and parametrizing that.
Margins of latent-variable models may satisfy constraints beyond 
conditional independence, including generalized equality constraints, 
such as the Verma constraint 
\citep[][Section~6.9]{Verma1990, Spirtes2000}, 
and inequality constraints, 
such as the instrumental inequalities of \citet{Pearl1995b}. 
The implementation based on the canonical directed acyclic graph 
fixes one particular latent structure whose latent projection is
$\mathcal{G}^p$, 
and may therefore fail to represent inequality
constraints implied by other latent-variable directed acyclic graphs
with the same latent projection.
Conversely, the implementation based on the maximal arid projection
works with the Gaussian nested Markov model, which 
does not, in general, capture inequality constraints
(in the discrete case, nested Markov models capture all
equality constraints though; \citealp{Evans2018}).
Thus, there may exist distributions in the Gaussian nested Markov model
of $\mathcal{G}^p$ that do not arise as observable margins of any such
latent-variable model \citep{Shpitser2014}.
This is not an obstacle for verifying candidate observational formulas 
for interventional distributions, since identification
of interventionals
in latent-variable 
causal directed acyclic graphs is characterized at the level 
of the latent projection \citep{Richardson2023}.
If the goal were instead
to verify (in-)equality constraints,
then the latent structure
would have to be specified explicitly.

\ifarxiv
    \appsection{Numerical evaluation and exact arithmetic}\label{app:exact}
\else
    \section{Numerical evaluation and exact arithmetic}\label{app:exact}
\fi

\zcref[S]{thm:as-correct-verifier} gives
an almost-sure guarantee for the falsification-based verifier under exact evaluation and 
absolutely continuous parameter sampling. 
This guarantee does not apply in floating-point implementations,
where exact equality is replaced by comparison 
up to a tolerance $\eta>0$, and a non-zero discrepancy may
be treated as zero.
To see this, consider the linear Gaussian
model on $T\to X_1\to\dots\to X_r\to Y$, with $r\in\mathbb{N}$,
and all pairwise conditionals being of the form $V\mid W=w \sim \mathcal{N}(w/2, 1)$ where $\operatorname{pa}(V) = \{W\}$.
Under $\operatorname{do}(T=t)$, the coefficient of $t$ in the 
interventional mean of $Y$ is $(1/2)^{r+1}$. 
Suppose that an incorrect candidate formula sets this coefficient 
to zero. Then, at $t=1$, the absolute discrepancy in the mean 
is $(1/2)^{r+1}$. With tolerance $\eta=10^{-6}$, 
this discrepancy is below the tolerance as soon as
$(1/2)^{r+1}<10^{-6}$, which first occurs at $r=19$. 
Thus, a tolerance-based falsifier may accept an invalid formula
simply because the discrepancy is numerically small.
Further decreasing the tolerance is not a principled solution, since
floating-point arithmetic cannot represent arbitrarily 
small positive numbers, and 
underflow and rounding limit what can be distinguished numerically.
For sufficiently large $r$, 
a non-zero discrepancy may therefore be treated as exactly zero.
Thus, a naive floating-point implementation 
does not inherit the almost-sure one-sided guarantee.
The measure-zero result rules out accidental exact agreement under 
exact evaluation and absolutely continuous parameter sampling,
but says nothing about non-zero discrepancies that are hidden by underflow or a fixed tolerance.
Treating this as a routine, numerical nuisance
would disconnect the implementation from the theorem.
A reliable falsifier must either decide the polynomial identity problem symbolically
or control and analyze the additional error 
of a deliberate numerical implementation.

In the linear Gaussian case, we consider an implementation based
on exact arithmetic
(and sampling parameters from a finite set of integers), which avoids floating-point error and allows 
us to bound the probability of this implementation falsely accepting a non-identifying
formula (\zcref[S]{thm:bound}).

Suppose that $\mathcal{G}$ has $n$ nodes.
We consider a parametric
family 
$\{\{p_{\theta_j}(v_j\mid\mathbf{v}_{\operatorname{pa}(V_j)})\}_{\theta_j\in\Theta_j}\}_{V_j\in\mathbf V}$
where, for all $j\in\{1,\dots,n\}$ and all 
$\theta_j=(\alpha_j,\{\beta_{ij}\}_{i:V_i\in\operatorname{pa}(V_j)},\sigma_j^2)\in\Theta_j\subseteq\mathbb{R}^{d_j}$, 
with $\sigma_j^2>0$,
\begin{equation*}
    p_{\theta_j}(v_j\mid\mathbf{v}_{\operatorname{pa}(V_j)})
    =
    \mathcal{N}\left(
    v_j;\
    \alpha_j+\sum_{i:V_i\in\operatorname{pa}(V_j)}\beta_{ij}v_i,\
    \sigma_j^2
    \right).
\end{equation*} 
The parameter vector
$\theta=(\theta_1,\dots,\theta_n)\in\Theta\subseteq\mathbb R^d$,
with $d=\sum_{j=1}^n d_j$,
then collects all intercepts, edge coefficients and conditional variances. 
Define $\alpha=(\alpha_1,\dots,\alpha_n)^\top$,
$\Omega=\operatorname{diag}(\sigma_1^2,\dots,\sigma_n^2)$,
and the $n\times n$ matrix $B$, where, for all $i,j\in\{1,\dots,n\}$,
$i\ne j$, $B_{ij}=\beta_{ij}$ if $V_i\in\operatorname{pa}(V_j)$, 
and $B_{ij}=0$ otherwise.
Then the product of all Gaussian conditionals induces 
a joint Gaussian distribution with mean vector and covariance matrix
\begin{equation}\label{eq:mean-covariance}
    \mu^\mathrm{full}(\theta)=(I-B^\top)^{-1}\alpha,
    \qquad
    \Sigma^\mathrm{full}(\theta)=(I-B^\top)^{-1}\Omega(I-B)^{-1}.
\end{equation}

For fixed $\theta\in\Theta$ and $\mathbf{t}\in\mathcal{X}_{\mathbf{T}}$,
both the distribution returned by the candidate observational
formula $\phi$ and the target interventional distribution are
Gaussian. Therefore, checking that, for all $\mathbf{t}\in\mathcal{X}_{\mathbf{T}}$,
$\llbracket \phi \rrbracket(p_{\theta,\mathbf{O}},\mathbf{t})=p_{\theta,\mathbf{Y}\mid\doop(\mathbf{T}=\mathbf{t})}$ $\mu_{\mathbf{Y}}$-almost everywhere is 
equivalent to checking equality of their mean vectors 
and covariance matrices; see also \zcref[S]{subsec:falsification-procedure}. 
For all $\theta\in\Theta$ and $\mathbf{t}\in\mathcal{X}_{\mathbf{T}}$,
let $\mu^\phi(\theta,\mathbf{t})$ and 
$\Sigma^\phi(\theta)$ denote the mean vector and covariance matrix 
obtained from the candidate formula $\phi$, and let 
$\mu^{\operatorname{do}}(\theta,\mathbf{t})$ and 
$\Sigma^{\operatorname{do}}(\theta)$ denote the corresponding
quantities for the target interventional distribution.
Since, for all $\theta\in\Theta$,
$\Sigma^\phi(\theta)$ and $\Sigma^{\operatorname{do}}(\theta)$ are
independent of $\mathbf{t}$,
and since, for all $\theta\in\Theta$, $\mu^\phi(\theta,\mathbf{t})$
and $\mu^{\operatorname{do}}(\theta,\mathbf{t})$ are affine functions
of $\mathbf{t}$, it is enough to compare the covariance matrices
and compare the mean vectors at $|\mathbf{T}|+1$ affinely independent
intervention values.

In a linear Gaussian model, the entries of these vectors 
and matrices are rational functions
of the model parameters. Therefore, after clearing denominators,
verification reduces to checking whether the resulting
polynomial differences are identically zero.
This problem is known as polynomial identity testing 
\citep{Shpilka2010}, defined as follows.
\begin{definition}[Polynomial identity testing]\label{def:PIT}
    Let $\mathbb{F}$ be a field, and let 
    $\mathbb{F}[x_1,\dots,x_s]$ denote the set of 
    polynomials in $x_1,\dots,x_s$ 
    with coefficients in $\mathbb{F}$.
    Given $g\in\mathbb{F}[x_1,\dots,x_s]$, the \emph{polynomial identity testing} problem is to decide whether $g$ is identically zero, that is, $g\equiv0$.
\end{definition}

Designing efficient deterministic algorithms for
polynomial identity testing is an open problem in algebraic 
complexity theory. 
A standard alternative to computationally expensive symbolic solutions is randomized evaluation:
choose a finite sampling set $A\subseteq\mathbb{F}$, sample 
$x^\prime=(x_1^\prime,\dots,x_s^\prime)\in A^s$, 
and evaluate $g(x^\prime)$ exactly.
If $g(x^\prime)\neq 0$, then necessarily $g\not\equiv0$.
If $g(x^\prime)=0$, then either $g\equiv 0$, 
or $g\not\equiv 0$ and $x^\prime$ lies in the zero set of $g$.
The following result, due to \citet{DeMillo1978}, 
\citet{Zippel1979}, and \citet{Schwartz1980}, 
and widely known as the Schwartz-Zippel lemma,
bounds the probability of this 
latter event.
\begin{theorem}[Schwartz-Zippel]\label{thm:schwartz-zippel}
    Let $A\subseteq\mathbb{F}$ be a non-empty finite set. 
    For every non-zero polynomial $g\in\mathbb{F}[x_1,\dots,x_s]$
    of total degree at most $D$, if 
    $x^\prime=(x_1^\prime,\dots,x_s^\prime)$
    is sampled uniformly from $A^s$, then 
    \begin{equation*}
        \mathbb{P}\bigl(g(x^\prime)=0\bigr)\leq \frac{D}{|A|}.
    \end{equation*}
\end{theorem}
Our implementation based on exact arithmetic therefore avoids
numerical issues that floating-point implementations would 
incur.
It does not, however, rely on a full symbolic decision
procedure for polynomial identity testing, which would be 
computationally expensive. Instead, we use a randomized procedure,
which replaces the absolutely continuous parameter sampling
in \zcref[S]{thm:as-correct-verifier} by finite-set sampling,
and therefore changes the  guarantee.
A non-zero polynomial may vanish at the sampled point,
despite exact arithmetic being used,
if the sampled point is exactly a zero of the polynomial.
With \zcref[S]{thm:schwartz-zippel}, we can bound the 
probability of this happening
and in turn of the proposed falsifier falsely accepting a non-identifying formula as valid.

For all $i\in\{1,\dots,d_1,\dots,d_1+d_2,\dots,d\}$, that is,
for each scalar parameter in $\theta$, 
let $A_i$ be a non-empty finite sampling set, 
and sample the corresponding parameter uniformly from $A_i$.
In our implementation,
intercepts $\alpha_j$ and edge coefficients $\beta_{ij}$
are sampled from $\{-M,\dots,M\}\setminus\{0\}$ and the
conditional variances are sampled from 
$\{1,\dots,2M\}$, with $M\in\mathbb{N}\setminus\{0\}$.
All subsequent operations are then evaluated
using exact rational arithmetic.

To apply \zcref[S]{thm:schwartz-zippel}, we need to bound 
the degree of the polynomial obtained from the difference
between the candidate and target mean or covariance entries,
after clearing denominators. This degree depends on the operations
used to evaluate the candidate formula, so we introduce the
notion of evaluation length.
We say that $\phi$ has evaluation length at most $c$ if,
for all $\theta\in\Theta$ and $\mathbf{t}\in\mathcal{X}_{\mathbf{T}}$,
each scalar entry of $\mu^\phi(\theta,\mathbf{t})$ and
$\Sigma^\phi(\theta)$ can be obtained from 
$\mu^\mathrm{full}(\theta)$ and $\Sigma^\mathrm{full}(\theta)$
using at most $c$ primitive operations.
Here, the evaluation length is not the number of terms appearing
in the displayed formula $\phi$.
Rather, the marginalizations, conditionals, products,
and quotients appearing in $\phi$ are translated into operations
on Gaussian means and covariances.
For each scalar entry of the resulting mean and covariance, the
evaluation length counts the operations needed to compute the entry.
Each scalar addition, subtraction, multiplication, or division
is counted as one primitive operation.
Even though Gaussian marginalization and conditioning act on 
vectors and matrices, we count them as primitive operations as well,
and account for their effect on the degrees of the resulting 
scalar rational expressions.
Primitive operations contribute separately
to the bound below, proven in \zcref[S]{app:proof-bound}.

\falseacceptancebound

By the non-identification premise,
there is at least one non-zero scalar entry in
$\mu^\phi(\theta^*,\mathbf{t}^*)-\mu^{\doop}(\theta^*,\mathbf{t}^*)$
or $\Sigma^\phi(\theta^*)-\Sigma^{\doop}(\theta^*)$.
There may be several such entries, and the implementation
compares all covariance entries and all mean entries at $|\mathbf{T}|+1$
intervention values.
However, for the probability bound, one non-zero discrepancy is enough.
Fix one such discrepancy.
If the falsifier accepts, then all checked discrepancies
vanish at the sampled parameter value; in particular,
this fixed discrepancy must also vanish at the sampled parameter value.
Therefore, the proof of \zcref[S]{thm:bound} uses that
the false-acceptance event is contained in the event that 
one non-zero polynomial, obtained from this fixed discrepancy after 
clearing denominators, evaluates to zero.

The following example shows how an observational formula
is translated into operations on the Gaussian mean and covariance,
and how these operations determine the resulting 
false-acceptance bound.

\begin{example}
[Computing the bound for the front-door formula]
\label{ex:front-door-bound}
Fix $\theta\in\Theta$ and $\mathbf t\in\mathcal X_{\mathbf T}$, and consider the front-door formula
\begin{equation*}
    \int p_{\theta, \mathbf{Z}\mid\mathbf{T}}(\mathbf z\mid\mathbf t)\int p_{\theta,\mathbf Y\mid \mathbf T,\mathbf Z}(\mathbf y\mid \mathbf t',\mathbf z) p_{\theta,\mathbf{T}}(\mathbf t')\,\mathrm d\mathbf t'\mathrm d\mathbf z,
\end{equation*}
where all distributions are Gaussian.
We count the primitive operations needed to compute one
scalar entry of the covariance matrix returned by this formula,
and then show that the same bound also controls the mean entries.
We use this count to compute the bound in \zcref[S]{eq:sharper-bound}.

We first translate the front-door formula into the corresponding
Gaussian distribution.
Let $\mathbf W\coloneqq(\mathbf T,\mathbf Z)$ and define
$\Gamma_{\mathbf Y\mathbf W}\coloneqq\Sigma_{\mathbf Y\mathbf W}\Sigma_{\mathbf W\mathbf W}^{-1}=[\,\Gamma_{\mathbf Y\mathbf T}\ \Gamma_{\mathbf Y\mathbf Z}\,],$
where $\Gamma_{\mathbf Y\mathbf T}$ and
$\Gamma_{\mathbf Y\mathbf Z}$ are the blocks of the regression
coefficient of $\mathbf Y$ on the joint vector
$(\mathbf T,\mathbf Z)$ corresponding to $\mathbf T$ and
$\mathbf Z$, respectively. Therefore, for all
$\mathbf t'\in\mathcal X_{\mathbf T}$ and
$\mathbf z\in\mathcal X_{\mathbf Z}$,
\begin{align*} 
    p_{\theta,\mathbf Y\mid \mathbf T,\mathbf Z}(\mathbf{y}\mid\mathbf{t}^\prime,\mathbf{z}) 
    &= \mathcal{N}(\mu_{\mathbf{Y}}+\Gamma_{\mathbf Y\mathbf W}(\mathbf{w}-\mu_{\mathbf{w}}),\Sigma_{\mathbf{Y}\mid\mathbf{W}})\\ 
    &= \mathcal{N}(\mu_{\mathbf{Y}}+\Gamma_{\mathbf{YT}}(\mathbf{t}^\prime-\mu_{\mathbf{T}})+\Gamma_{\mathbf{YZ}}(\mathbf{z}-\mu_{\mathbf{z}}),\Sigma_{\mathbf{Y}\mid\mathbf{T},\mathbf{Z}}). \end{align*}
We use the following Gaussian affine-integration identity. For all
$r,s\in\mathbb N$, all $a\in\mathbb R^r$,
$B\in\mathbb R^{r\times s}$, $m\in\mathbb R^s$, all positive
definite $S\in\mathbb R^{r\times r}$ and
$V\in\mathbb R^{s\times s}$, and all
$\mathbf y\in\mathbb R^r$,
\[
    \int_{\mathbb R^s}
    \mathcal N(\mathbf y; a+B\mathbf x,S)
    \mathcal N(\mathbf x;m,V)\,d\mathbf x
    =
    \mathcal N(\mathbf y; a+Bm,S+BVB^\top).
\]
Since
$p_{\theta,\mathbf{T}}(\mathbf t')=
\mathcal N(\mu_{\mathbf T},\Sigma_{\mathbf T\mathbf T})$,
applying this identity to the inner integral gives, for all
$\mathbf z\in\mathcal X_{\mathbf Z}$,
\[
    \int p_{\theta,\mathbf Y\mid \mathbf T,\mathbf Z}(\mathbf y\mid \mathbf t',\mathbf z)
    p_{\theta,\mathbf{T}}(\mathbf t')\,\mathrm d\mathbf t'
    =
    \mathcal N\!\left(
    \mu_{\mathbf Y}
    +
    \Gamma_{\mathbf Y\mathbf Z}(\mathbf z-\mu_{\mathbf Z}),
    \Sigma_{\mathbf Y\mid\mathbf T,\mathbf Z}
    +
    \Gamma_{\mathbf Y\mathbf T}
    \Sigma_{\mathbf T\mathbf T}
    \Gamma_{\mathbf Y\mathbf T}^{\top}
    \right).
\]
We now evaluate the outer integral. Since
\[
    p_{\theta, \mathbf{Z}\mid\mathbf{T}}(\mathbf z\mid\mathbf t)
    =
    \mathcal N\!\left(
    \mu_{\mathbf Z}
    +
    \Gamma_{\mathbf Z\mathbf T}(\mathbf t-\mu_{\mathbf T}),
    \Sigma_{\mathbf Z\mid\mathbf T}
    \right),
    \qquad
    \Gamma_{\mathbf Z\mathbf T}
    \coloneqq
    \Sigma_{\mathbf Z\mathbf T}\Sigma_{\mathbf T\mathbf T}^{-1},
\]
applying the same affine-integration identity to the outer
integral gives
\[
    \llbracket \phi_{\mathbf Z}^{\mathrm{fd}}\rrbracket
    (p_{\theta,\mathbf O},\mathbf t)
    =
    \mathcal N(\mu^{\mathrm{fd}}(\theta,\mathbf t),
    \Sigma^{\mathrm{fd}}(\theta)),
\]
where
\begin{align*}
    \mu^{\mathrm{fd}}(\theta,\mathbf t)
    &=
    \mu_{\mathbf Y}
    +
    \Gamma_{\mathbf Y\mathbf Z}
    \Gamma_{\mathbf Z\mathbf T}
    (\mathbf t-\mu_{\mathbf T}),\\
    \Sigma^{\mathrm{fd}}(\theta)
    &=
    \Sigma_{\mathbf Y\mid\mathbf T,\mathbf Z}
    +
    \Gamma_{\mathbf Y\mathbf T}
    \Sigma_{\mathbf T\mathbf T}
    \Gamma_{\mathbf Y\mathbf T}^{\top}
    +
    \Gamma_{\mathbf Y\mathbf Z}
    \Sigma_{\mathbf Z\mid\mathbf T}
    \Gamma_{\mathbf Y\mathbf Z}^{\top}.
\end{align*}

We now count the primitive operations needed to compute one scalar
entry of $\Sigma^{\mathrm{fd}}(\theta)$. Let
$\tau\coloneqq|\mathbf T|$ and $\zeta\coloneqq|\mathbf Z|$.
For all $i,j\in\{1,\dots,|\mathbf Y|\}$, let $a_i$ and $a_j$
be the $i$-th and $j$-th rows of
$\Gamma_{\mathbf Y\mathbf T}$, and let $b_i$ and $b_j$ be the
$i$-th and $j$-th rows of $\Gamma_{\mathbf Y\mathbf Z}$.
Then
\[
    \Sigma^{\mathrm{fd}}_{ij}(\theta)
    =
    (\Sigma_{\mathbf Y\mid\mathbf T,\mathbf Z})_{ij}
    +
    a_i^\top\Sigma_{\mathbf T\mathbf T}a_j
    +
    b_i^\top\Sigma_{\mathbf Z\mid\mathbf T}b_j.
\]

For $r\in\mathbb N$, a scalar bilinear form
$u^\top Mv$, with $u,v\in\mathbb R^r$ and
$M\in\mathbb R^{r\times r}$, can be computed by first computing
$Mv$ and then multiplying by $u^\top$. Computing $Mv$ requires
$r^2$ scalar multiplications and $r(r-1)$ scalar additions.
Multiplying the result by $u^\top$ requires $r$ scalar
multiplications and $r-1$ scalar additions. Hence one such
bilinear form requires
\[
    r^2+r(r-1)+r+(r-1)=2r^2+r-1
\]
scalar primitive operations.
Therefore, for all $i,j\in\{1,\dots,|\mathbf Y|\}$,
the term $a_i^\top\Sigma_{\mathbf T\mathbf T}a_j$ requires
$2\tau^2+\tau-1$ scalar primitive operations, and the term
$b_i^\top\Sigma_{\mathbf Z\mid\mathbf T}b_j$ requires
$2\zeta^2+\zeta-1$ scalar primitive operations. Adding the
three scalar terms in $\Sigma^{\mathrm{fd}}_{ij}(\theta)$
requires two additional scalar additions. Thus, after the
two Gaussian conditioning operations used to obtain
$\Sigma_{\mathbf{Y}\mid\mathbf{T},\mathbf{Z}}$ and $\Sigma_{\mathbf{Z}\mid\mathbf{T}}$, the number of
scalar primitive operations, also accounting for the final $2$ additions,
needed for one covariance entry is
\[
    (2\tau^2+\tau-1)+(2\zeta^2+\zeta-1)+2
    =
    2\tau^2+\tau+2\zeta^2+\zeta.
\]

We now turn to the primitive operations needed to compute one mean
entry. For all $i\in\{1,\dots,|\mathbf Y|\}$, let $b_i$ be the
$i$-th row of $\Gamma_{\mathbf Y\mathbf Z}$. Then
\[
    \mu^{\mathrm{fd}}_i(\theta,\mathbf t)
    =
    \mu_{\mathbf Y,i}
    +
    b_i^\top
    \Gamma_{\mathbf Z\mathbf T}
    (\mathbf t-\mu_{\mathbf T}).
\]
Computing $\mathbf t-\mu_{\mathbf T}$ requires $\tau$ scalar
subtractions. Multiplying this vector by
$\Gamma_{\mathbf Z\mathbf T}$ requires $\zeta\tau$ scalar
multiplications and $\zeta(\tau-1)$ scalar additions. Multiplying
the result by $b_i^\top$ requires $\zeta$ scalar multiplications
and $\zeta-1$ scalar additions, and adding
$\mu_{\mathbf Y,i}$ requires one additional scalar addition.
Hence one mean entry requires
\[
    \tau+\zeta\tau+\zeta(\tau-1)+\zeta+(\zeta-1)+1
    =
    \tau+2\zeta\tau+\zeta
\]
scalar primitive operations. Since
$\tau+2\zeta\tau+\zeta\le2\tau^2+\tau+2\zeta^2+\zeta$,
the covariance-entry count also upper bounds the number of scalar
operations needed to compute one mean entry.

By \zcref[S]{eq:gamma}, each primitive operation contributes
a multiplicative factor describing how much it can increase 
the current degree bound.
Gaussian marginalizations contribute a factor $1$, 
the two Gaussian conditioning operations
contribute factors $n+1$ each, and each scalar operation
contributes factor $2$. Thus, for both the covariance and mean
entries of the front-door formula,
\[
    \prod_{l=1}^c\gamma_l
    \le
    (n+1)^2
    2^{\,2\tau^2+\tau+2\zeta^2+\zeta}.
\]
If $\tau=\zeta=1$, $n=10$ and 
$\min_{1\le i\le d}|A_i|=2^{64}-1$, the upper 
bound in \zcref[S]{eq:sharper-bound} then is approximately 
$7.98\times 10^{-15}$. 
\end{example}

The bound in \zcref[S]{eq:sharper-bound} is valid for the 
Gaussian parametrization using mean vector and covariance matrix.
Other parametrizations can lead to different bounds, because 
the primitive operations may have different algebraic cost.
For example, in canonical form, one represents a Gaussian
by its precision matrix and information vector.
In this parametrization, conditioning is comparatively 
cheap, while marginalization is more expensive.

We can reduce the bound in \zcref[S]{eq:sharper-bound} by increasing 
the cardinalities of the sampling sets. 
This is possible with arbitrary precision integer sampling,
but it may increase the cost of exact
arithmetic, because larger sampled integers can lead to rational
computations with larger bit lengths. 
In our examples, we did not observe
a substantial slowdown, but we also provide a floating-point
implementation for cases in which exact evaluation becomes
computationally expensive.
The bound can also be reduced 
by repeating the test independently. Let $\delta$ be this bound, 
truncated at $1$. Then one run falsely accepts a
non-identifying formula with probability at most $\delta$, and $K$
independent repetitions falsely accept it in all runs 
with probability at most $\delta^K$.

\ifarxiv
    \appsection{The front-door criterion}\label{app:frontdoor}
\else
    \section{The front-door criterion}\label{app:frontdoor}
\fi

The \emph{front-door criterion} \citep[][Section~3.2]{Pearl1995}
gives graphical conditions under which a candidate set 
$\mathbf Z\subseteq\mathbf{O}\setminus(\mathbf{T}\sqcup\mathbf{Y})$
allows to identify $\mathbf{Y}\mid\doop(\mathbf{T})$
with the corresponding front-door formula:
\begin{equation}\label{eq:front-door-formula}
    \sum_{\mathbf{z}}p(\mathbf{z}\mid \mathbf t)\sum_{\mathbf t'}p(\mathbf y\mid \mathbf t', \mathbf{z})\,p(\mathbf t').
\end{equation}
In particular, a set of variables $\mathbf{Z}$ 
satisfies the front-door criterion if:
\begin{enumerate}[label=(\roman*)]
    \item $\mathbf{Z}$ intercepts all directed paths from $\mathbf T$ to $\mathbf Y$;
    \item there is no unblocked back-door path from $\mathbf T$ to $\mathbf{Z}$; 
    \item all back-door paths from $\mathbf{Z}$ to $\mathbf Y$ are blocked by $\mathbf T$.
\end{enumerate}
Whenever these conditions hold, the corresponding front-door formula 
is identifying for $\mathbf{Y}\mid\doop(\mathbf{T})$.

\subsection{Identification by a front-door formula with a set not satisfying the front-door criterion}\label{app:proof-front-door}

Consider the acyclic directed mixed graph in 
\zcref[S]{fig:front-door}. 
We show that, even though $\{M,A\}$ does not satisfy the 
front-door criterion, the corresponding front-door formula 
$$\phi_{\{M,A\}}^\mathrm{fd}=\sum_{m,a}p(m,a\mid t)\sum_{t^\prime}p(y\mid t^\prime,m,a)p(t^\prime)$$
is identifying for $Y\mid\doop(T)$
in the corresponding canonical directed acyclic graph,
obtained by replacing the bidirected edge with a node $U$ 
such that $T\leftarrow U\rightarrow Y$.
The same holds for $\{M,C\}$ and the proof is analogous.

By the truncated factorization formula 
\citep[][Section~1.3]{Pearl2009},
\begin{equation*}
    p(y\mid\doop(t))=\sum_{m,a,u,b,c}p(u)p(b)p(m\mid t)p(a\mid b)p(c\mid b)p(y\mid m,a,c,u).
\end{equation*}

Since $M$ and $A$ are $d$-separated by $T$ and $A$ and $T$ are $d$-separated by the empty set,
\begin{equation}
    p(m, a\mid t)=p(m\mid t)p(a\mid t)=p(m\mid t)p(a).
\end{equation}
Moreover,
\begin{align}
    \sum_{t^\prime}p(y\mid t^\prime,m,a)p(t^\prime)
    &=\sum_{t^\prime}\sum_{u,b,c}p(y,u,b,c\mid t^\prime,m,a)p(t^\prime)\\
    &=\sum_{t^\prime}\sum_{u,b,c}p(y\mid t^\prime,m,a,u,b,c)p(u,b,c\mid t^\prime,m,a)p(t^\prime)\\
    &\overset{\mathclap{(1)}}{=}\sum_{t^\prime}\sum_{u,b,c}p(y\mid m,a,u,c)\frac{p(u)p(b)p(t^\prime\mid u)p(m\mid t^\prime)p(a\mid b)p(c\mid b)}{p(t^\prime)p(m\mid t^\prime)p(a)}p(t^\prime)\\
    &=\sum_{t^\prime}\sum_{u,b,c}p(y\mid m,a,u,c)p(u)p(t^\prime\mid u)p(b)\frac{p(a\mid b)}{p(a)}p(c\mid b)\\
    &\overset{\mathclap{(2)}}{=}\sum_{u,b,c}p(y\mid m,a,u,c)p(u)p(b\mid a)p(c\mid b),
\end{align}
where in $(1)$ we used that $Y$ is $d$-separated from $T$ and $B$ 
given $\{M,A,U,C\}$, and in $(2)$ we used that, for all $u$, 
$\sum_{t^\prime}p(t^\prime\mid u)=1$ and $p(b)p(a\mid b)/p(a)=p(b\mid a)$.
Therefore,
\begin{align}
    \phi_{\{M,A\}}^\mathrm{fd}
    &=
    \sum_{m,a,u,b,c}p(m\mid t)p(a)p(y\mid m,a,u,c)p(u)p(b\mid a)p(c\mid b)\\
    &\overset{\mathclap{(3)}}{=}
    \sum_{m,a,u,b,c}p(m\mid t)p(u)p(b)p(a\mid b)p(c\mid b)p(y\mid m,a,u,c)\\
    &=p(y\mid\doop(t)),
\end{align}
where in $(3)$ we used that $p(a)p(b\mid a)=p(b)p(a\mid b)$.

\ifarxiv
    \appsection{Recovering all identifying formulas in a finite class}\label{app:sound-and-complete}
\else
    \section{Recovering all identifying formulas in a finite class}\label{app:sound-and-complete}
\fi

The gateway test described in
\ifsupplementary
    Section~5
\else
    \zcref[S]{sec:gateway} 
\fi
is sound and exhaustively complete relative to 
$\Phi_{\mathbf{Y},\mathbf{T}}^\mathrm{fd}$.
The underlying idea applies more generally to any finite
class of candidate observational formulas:
enumerate the formulas, verify each one,
and return exactly those that are verified to be identifying for 
$\mathbf{Y}\mid\doop(\mathbf{T})$.
More precisely, let 
$\widetilde\Phi_{\mathbf{Y},\mathbf{T}}\subseteq\Phi_{\mathbf{Y},\mathbf{T}}$
be a finite class of observational formulas.
Each $\phi\in\widetilde\Phi_{\mathbf{Y},\mathbf{T}}$ is verified
in turn, and passes the test if it is
verified as identifying.
With an exact verifier, this procedure returns all and 
only the identifying formulas in
$\widetilde\Phi_{\mathbf{Y},\mathbf{T}}$.
If one instead uses the falsifier from \zcref[S]{sec:hiprof},
then the procedure is almost-surely sound while it is exhaustively 
complete relative to $\widetilde\Phi_{\mathbf{Y},\mathbf{T}}$
for the conditional exponential family chosen in the falsification routine.

Graphical criteria that are sound and exhaustively complete relative
to $\widetilde\Phi_{\mathbf{Y},\mathbf{T}}$ may therefore be viewed
as algorithmic shortcuts to this exhaustive procedure.
However, this procedure remains applicable even when the 
graphical criterion is sound but not complete with respect
to the formula class
(as is the case for the front-door criterion with respect 
to the class of front-door formulas; see 
\ifsupplementary
    Section~5),
\else
    \zcref[S]{sec:gateway}),
\fi
or when no such graphical criterion is available.
For instance, as discussed in \zcref[S]{app:adjustment},
sound and exhaustively complete graphical criteria relative to
the class of adjustment formulas 
$\Phi_{\mathbf{Y},\mathbf{T}}^{\mathrm{adj}}$
defined in \zcref[S]{eq:adjustment-class} are 
available for directed  acyclic graphs, 
maximal ancestral graphs, or their equivalence classes. 
For acyclic directed mixed graphs, however, no graphical
criterion is currently known to be both sound and 
exhaustively complete relative to
$\Phi_{\mathbf{Y},\mathbf{T}}^{\mathrm{adj}}$. 
Verification enables us to fill this gap.

\end{document}